\documentclass[accepted]{uai2023} % after acceptance, for a revised
\usepackage[american]{babel}
\usepackage{natbib} % has a nice set of citation styles and commands
\usepackage{mathtools} % amsmath with fixes and additions
\usepackage{booktabs} % commands to create good-looking tables
\usepackage{tikz} % nice language for creating drawings and diagrams

\usepackage{xcolor}
\usepackage{algorithm}
\usepackage[noend]{algpseudocode}
\usepackage{caption}
\usepackage{subcaption}
\usepackage{graphicx}
\usepackage[edges]{forest}
\usepackage{amsfonts}
\usepackage{amsmath}
\usepackage{amsthm}
\usepackage[T1]{fontenc}

 \newenvironment{proofof}[1]{\noindent{\bf Proof of #1:}}{$\qed$\par}

\usepackage{comment}
\newtheorem{thm}{\protect\theoremname}
\newtheorem{lem}[thm]{\protect\lemmaname}
\newtheorem{cor}[thm]{\protect\corollaryname}
\newtheorem{defn}[thm]{\protect\definitionname}
\newtheorem{rem}{\protect\remarkname}

\makeatother

\providecommand{\claimname}{Claim}
\providecommand{\lemmaname}{Lemma}
\providecommand{\propositionname}{Proposition}
\providecommand{\theoremname}{Theorem}
\providecommand{\corollaryname}{Corollary}
\providecommand{\definitionname}{Definition}
\providecommand{\assumptionname}{Assumption}
\providecommand{\remarkname}{Remark}
\providecommand{\factname}{Fact}
\providecommand{\obsname}{Observation}
\providecommand{\questionname}{Question}

\newcommand{\wt}{\widetilde}

\newcommand{\KD}{\textsc{KD}}

\usepackage{xr}

\title{Differentially Private Synthetic Data Using KD-Trees}

\author[1]{\href{mailto:<eleonora.kreacic@jpmchase.com>?Subject=Your UAI 2023 paper}{Eleonora~Krea\v{c}i\'{c}}}
\author[1]{Navid~Nouri}
\author[1]{Vamsi~K.~Potluru}
\author[1]{Tucker~Balch}
\author[1]{Manuela~Veloso}
\affil[1]{%
    J.P. Morgan AI Research
}
\begin{document}
\maketitle

\begin{abstract}
    Creation of a synthetic dataset that faithfully represents the data distribution and simultaneously preserves privacy is a major research challenge. Many space partitioning based approaches have emerged in recent years for answering statistical queries in a differentially private manner. However, for synthetic data generation problem, recent research has been mainly focused on deep generative models. In contrast, we exploit space partitioning techniques together with noise perturbation and thus achieve intuitive and transparent algorithms. We propose both data independent and data dependent algorithms for $\epsilon$-differentially private synthetic data generation whose kernel density resembles that of the real dataset. Additionally, we provide theoretical results on the utility-privacy trade-offs and show how our data dependent approach overcomes the curse of dimensionality and leads to a scalable algorithm. We show empirical utility improvements over the prior work, and discuss performance of our algorithm on a downstream classification task on a real dataset. 
\end{abstract}

\section{Introduction}

Publishing data of a highly sensitive nature in domains of finance or health, carries a risk of compromising privacy of individuals and therefore a breach of privacy regulations (e.g. HIPPA, FCRA, GDPR). This limitation can be potentially circumvented by the use of synthetic data. However, synthetic data per se is not inherently private \citep{jordon2022synthetic}. In this paper, we study the problem of publishing a synthetic dataset that faithfully represents the original data whilst at the same time does not comproimise privacy of individuals in the original. 

\begin{figure}[h]
    \centering
\includegraphics[width=0.3\textwidth]{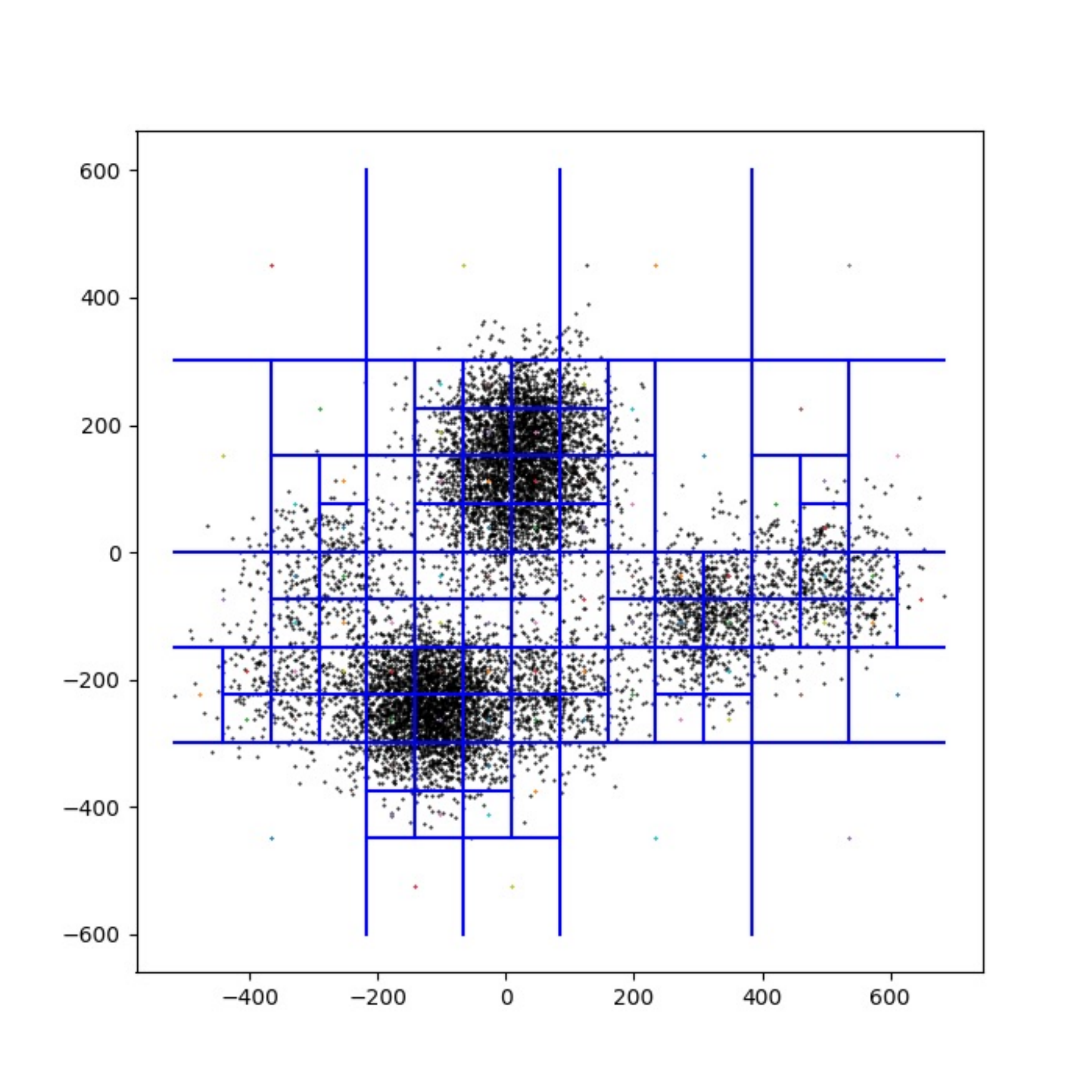}
    \caption{Data dependent partitioning in $\mathbb{R}^2$. Our data dependent algorithm achieves more refined partitioning in densely populated areas and thus better utility of DP synthetic data.} %of private synthetic data.}
    \label{fig:binning}
\end{figure}
Preserving privacy of individuals while publishing a dataset for public use is a known challenge \citep{DBLP:conf/icml/BalogTS18}. A de facto standard for privacy is differential privacy (DP), which is widely used in the literature and in practice. 
Many existing works aim to preserve quality of differentially private answers for a certain class of queries \citep{DBLP:conf/kdd/MohammedCFY11, DBLP:journals/tdp/XiaoXFGL14,NIPS2012_208e43f0}; see \cite{7911185} for a survey. However, a more general problem studies the release of a differentially private synthetic dataset that can be used for downstream tasks without additional privacy leaks. Recent work mainly applied generative adversarial networks (GAN) \citep{DBLP:conf/nips/GoodfellowPMXWOCB14} and utilized divergence metrics such as Jensen-Shannon divergence and Wasserstein distance as a metric of quality to compare synthetic and original datasets \citep{DBLP:journals/pvldb/ParkMGJPK18,DBLP:conf/cvpr/Torkzadehmahani19,DBLP:conf/iclr/JordonYS19a,DBLP:journals/corr/abs-1802-06739,DBLP:conf/sec/FrigerioOGD19, https://doi.org/10.48550/arxiv.1610.05755, https://doi.org/10.48550/arxiv.1802.08908}. Another class of utility metrics is based on kernels e.g., distance in reproducing kernel Hilbert space (RKHS) or similarly maximum mean discrepancy (MMD), see Section~\ref{sect: KME and RKHS distance} for a detailed discussion. The advantage of kernel based metrics is that they compare two probability measures in terms of all possible moments \citep{DBLP:conf/aistats/HarderAP21} and thus capture a wide class of statistical properties of the dataset.

We propose algorithms that for an input dataset output its differentially private synthetic counterpart. The key techniques used in this work are space partitioning~\citep{ram2013space} and noisy perturbation. Similar ideas have been used in the literature for designing differentially private statistical query algorithms and histogram release applications, e.g., \citep{DBLP:conf/icde/QardajiYL13,DBLP:journals/tdp/XiaoXFGL14,DBLP:conf/sigmod/ZhangXX16}. However, naive implementation of space partitioning techniques for DP dataset release problem suffers from the curse of dimensionality. %in high dimensions. 
We show how our data-dependent approach solves this issue and also leads to scalable algorithms. As our method is based on the simple idea of space partitioning, the advantage in comparison to the state of the art generative models is interpretability of our algorithms.

\subsection{Related work}
Prior work which consider similar settings to ours is relatively sparse. The closest are \cite{DBLP:conf/icml/BalogTS18} and \cite{DBLP:conf/aistats/HarderAP21}, although the latter is more focused on labeled data as well as image data. 
Both of these papers aim to release a dataset that is close to the original one in terms of kernel RKHS distance while preserving privacy. \cite{DBLP:conf/icml/BalogTS18} presents two algorithms. Their first algorithm either needs a small fraction of the dataset to be publicly available, or reweights a sampled set of points from the support of the underlying distribution. The requirement on the public input data represents a limitation which we overcome.
Their second algorithm is based on random features and an iterative gradient based optimization procedure to apply reduced set method which has the twin issues of being slow in high dimensions and also suffers from the lack of interpretability. \cite{DBLP:conf/aistats/HarderAP21} improved upon this work and achieved better results, and in particular for lower dimensions, say $5$, but not when the dimensions are much higher. Their approach relies on deep generative model to minimize MMD, and thus again suffers the lack of interpretability. None of these works provides theoretical results on the utility-privacy trade-off. Our algorithms are transparent which enables us to provide theoretical guarantees on utility. Our data dependent algorithm achieves good utility in high dimensions.

\subsection{Our contribution}
Inspired by locality sensitive hash functions and nearest neighbor search algorithms introduced by \cite{DBLP:conf/stoc/IndykM98}, our approach is based on space partitioning schemes and in particular KD-trees \cite{KDtree}. Our goal is to output a synthetic dataset that imposes similar kernel density on the space as the input data, and does not compromise privacy of the input. Intuitively, if we partition the space into small sections (bins) and preserve the ratio of points
after noise addition (for the purpose of privacy) we will approximately preserve the kernel density. Based on this idea, we introduce two algorithms that yield $\epsilon$-differentially private synthetic data outputs. Our data independent approach, Algorithm $1$, implements a naive version of the idea and naturally suffers from the curse of dimensionality. Our data dependent algorithm, Algorithm $2$, overcomes this and achieves better utility. Our main contributions are:
\begin{itemize}
    \item We provide an upper bound on utility loss of our data independent algorithm for a general setting 
    of unknown distribution of input data (Theorem \ref{thm:general-1} and Theorem \ref{thm:general-2}). We improve the bound 
    for a special case of input data from a mixture of Gaussians in $\mathbb{R}^{d}$ (Theorem \ref{thm one gaussian tradeoff}).

    \item 
    Our data dependent algorithm achieves smaller number of empty bins and more refined partitioning in densely populated areas which yields 
    better utility. We overcome the curse of dimensionality using an implicit
    sampling of empty bins (Theorem \ref{thm adaptive binning DP} and Theorem \ref{thm data dependant algo DP}).
    
    \item Unlike previous approaches \citep{DBLP:conf/icml/BalogTS18,DBLP:conf/aistats/HarderAP21}, we do not rely on black box methods, and thus achieve interpretability which is important in many practical settings. 
    In addition, we do not require a fraction of input dataset to be public.
    
    \item In Section \ref{sect experiments} we show how: (i) our algorithms outperform algorithms by \cite{DBLP:conf/icml/BalogTS18}, (ii) our data dependent algorithm overcomes curse of dimensionality, (iii) empirical utility loss compares to our theoretical bound from Theorem \ref{thm one gaussian tradeoff}, (iv) performance of our algorithms on downstream binary classification task compares with \cite{DBLP:conf/aistats/HarderAP21}.
\end{itemize}

\subsection{Background}
We give an introduction to MMD and differential privacy.
\subsubsection{Kernel Density Estimates and RKHS distance} \label{sect: KME and RKHS distance}

For a (unweighted) dataset $P\subset{\mathbb{R}^{d}}$ and a kernel $K:{\mathbb{R}^{d}}\times{\mathbb{R}^{d}}\to \mathbb{R}$, the kernel density (KD) with respect to $P$ is defined at any point $x\in\mathbb{R}^{d}$ as $\text{KD}^{K}_{P}(x) = \frac{1}{\vert P \vert}\sum_{p\in P}K(x,p)$. If $P$ is equipped with weights such that $\sum_{p\in P}w_{p}=1$, then kernel density is given by
$\text{KD}^{K}_{P}(x) = \sum_{p\in P}w_{p}K(x,p)$. For ease of notation, we will often write $\text{KD}_{P}(\cdot)$ instead of $\text{KD}^{K}_{P}(\cdot)$. For the two datasets $P$ and $Q$, $\ell_{\infty}$ distance of two KDs is defined as $\Vert \text{KD}_{P} - \text{KD}_{Q} \Vert_{\infty} = \sup_{x \in \mathbb{R}^{d}} \vert \text{KD}_{P}(x)- \text{KD}_{Q}(x) \vert$. If $K$ is positive definite, then $K(p,x)$ can be represented as an inner product in RKHS $\mathcal{H}_{K}$. That is, there is $\phi_{K}:\mathbb{R}^{d}\to\mathcal{H}_{K}$ such that $\phi_{K}(x) = K(x, \cdot)$. For a positive definite kernel $K$, if $\phi_{K}$ is injective then MMD is given by
\begin{align}
    \text{MMD}(P, Q) = \sqrt{\kappa(P, P)+\kappa(Q,Q)-2\kappa(P,Q)}
\end{align}
where $\kappa(P,Q) = \frac{1}{\vert P \vert }\frac{1}{\vert Q \vert }\sum_{p \in P}\sum_{q \in Q}K(p,q)$
represents a kernel metric between two datasets.
It is possible to convert between bounds on $\ell_{\infty}$-distance of KDs and MMD. More precisely, $\text{MMD}(P, Q)\leq \epsilon$ implies $\Vert \text{KD}_{P} - \text{KD}_{Q} \Vert_{\infty} \leq\epsilon$, and also $\Vert \text{KD}_{P} - \text{KD}_{Q} \Vert_{\infty}\leq \epsilon$  implies $\text{MMD}(P, Q) \leq \sqrt{2\epsilon}$. 
Further details can be found in \cite{DBLP:journals/dcg/PhillipsT20} and \cite{DBLP:conf/aistats/HarderAP21}. We present theoretical analysis on utility guarantees in terms of bounds on KD distance, whilst in experiments we rely on MMD due to the ease of computation.

\subsubsection{Differential Privacy} \label{sect: intro DP}
Differential privacy (DP) \citep{dwork2006calibrating} has become a de facto standard to quantify 
privacy leakage. %from publishing functions of a dataset. 
It provides theoretical guarantees that potential adversary with the knowledge of the output is not able to distinguish whether a particular individual was present in the input dataset. %or not.
\begin{defn}\citep{dwork2014algorithmic}
A randomized mechanism $\mathcal{M}:\mathcal{X}^n \rightarrow \mathcal{Y}$ is $\epsilon$-differentially private if for any two datasets $\mathcal{D}, \mathcal{D'}\in \mathcal {X}^n$ that differ in only one entry,
we have
\begin{align}
\forall \mathcal{C}\subseteq \mathcal {Y},~\quad \mathbb{P}(\mathcal{M} (\mathcal{D}) \in \mathcal{C}) \leq e^{\epsilon} \mathbb{P}(\mathcal{M}(\mathcal{D'}) \in \mathcal{C}).
\end{align}
\end{defn}
Standard way to achieve $\epsilon$-DP is to employ Laplace mechanism, i.e., add Laplace noise to the output. More precisely, for a function $f$ computed on sensitive data $\mathcal{D}$, we introduce
$
\mathcal{M}_{Lap}(\mathcal{D}, f(.), \epsilon) = f(\mathcal{D} ) + \text{Lap}(0, \Delta_f/\epsilon),
$
where $\Delta_f = \max_{\mathcal{D} ,\mathcal{D'}}  \|f(\mathcal{D}) - f( \mathcal{D'} )\|_1$  is the $\ell_1$ sensitivity of $f$ with respect to change of a single entry in the dataset ($\mathcal{D},\mathcal{D'}\in \mathcal{ X}^n$ are two neighboring datasets, i.e.,  that they differ in only one entry) and $\text{Lap}$ denotes Laplace distribution parametrized by the mean and scale. 
\textbf{Post-processing property of $\epsilon$-DP} guarantees that composition of any data independent function with the output of $\epsilon$-DP mechanism is also $\epsilon$-DP, i.e., it does not incur additional privacy leaks \citep{dwork2014algorithmic}. This means that differentially private synthetic data can be safely used for downstream tasks. \textbf{Composition of $\epsilon$-DP} guarantees that combination of $\epsilon_{1}$-DP algorithm $\mathcal{M}_{1}$ and $\epsilon_{2}$-DP algorithm $\mathcal{M}_{2}$ defined by $\mathcal{M}_{1,2} = (\mathcal{M}_{1}, \mathcal{M}_{2})$ is $(\epsilon_{1}+\epsilon_{2})$-DP \citep{dwork2014algorithmic}.

\section{Problem formulation}

We are given a multidimensional numerical dataset $P=\{p_1,p_2,\ldots,p_n\}$ of $n$ records in $\mathbb{R}^{d}$. Our task is to design a differentially private algorithm that outputs (possibly weighted) dataset $Q = \{ (q_{1}, w_{1}), \dots, (q_{m}, w_{m})\}$ where $q_{i}\in \mathbb{R}^{d}$, $w_{i}\in \mathbb{R^{+}}$, such that for any $x\in \mathbb{R}^d$
    \begin{align*}
        \KD_P^K(x) \approx \KD_Q^K(x),
    \end{align*}
where $K$ is some positive definite kernel.
The closeness of $\KD_P$ and $\KD_Q$ in $\ell_\infty$-distance implies that relying on $Q$ instead of $P$ leads consistent estimation of population statistics of original dataset $P$ (see \cite{DBLP:conf/icml/BalogTS18} for discussion) i.e., synthetic dataset $Q$ faithfully represents the original $P$. In other words, both $\ell_\infty$-distance of KDs and MMD represent good utility measures when evaluating quality of synthetic datasets. In the rest of the paper, for simplicity of presentation we focus on Gaussian kernel $K(x,p)=e^{-\frac{\Vert x- p\Vert^{2}_{2}}{2\sigma^2}}$. It is however straightforward to adapt our analysis to a wider class of kernels.

\section{Our Algorithms}
We propose two algorithms for synthetic data generation in the next sections: (i) Data independent (ii) Data dependent.
\subsection{Data independent}
In this section, we present data independent algorithm for synthetic dataset release with DP guarantees. Inspired by the widely used idea of space partitioning, we want to partition the space into a number of bins, e.g., $J$ cubes of width $w$. Then we count the number of points inside each bin  and present these counts on a $J$-dimensional vector. 
Any single data point can affect this vector at most by a constant in terms of $\ell_1$ distance. In other words, it has a bounded $\ell_1$ sensitivity with respect to any two neighbouring datasets (see Section \ref{sect: intro DP}). Thus, we can employ Laplace mechanism in order to achieve $\epsilon$-DP. Bins with the noisy count below input threshold $t$ will be removed i.e., filtered out. The algorithm outputs the dataset consisting of centers of the bins that survived filtering step and the corresponding noisy point counts. See Algorithm~\ref{alg:rounding-new}.

\begin{algorithm}
	\caption{Data independent binning}  
	\label{alg:rounding-new} 
	\begin{algorithmic}[1]

		\Procedure{$\textsc{DataIndependent}(P,\epsilon , t, w)$}{}\\ \Comment{$P$ is the original dataset, $\epsilon$ is privacy budget}\\
  \Comment{$t$ is filtering threshold, $w$ is bin width}
		\State $R \gets $ the edge length of the axes aligned hypercube that encompasses $P$
		%\State \textcolor{red}{Choose a bin width $w$, independently of $P$}
		\State Apply binning using bins of width $w$ 
		\State $J \gets \left({R}/{w}\right)^d$ \Comment{Number of bins}
		\State $\mathbf{v} \in \mathbb{R}^J \gets$ vector of point counts bin by bin \label{line:v}
		\State $c_1, c_2,\ldots, c_J \gets$ centers of bins  
         \label{line:c}
		\State $\wt{\mathbf{v}} \gets \mathbf{v} + Lap(\frac{2}{\epsilon}I_{J\times J})$ \label{line:laplace-noise} \Comment{Noisy point counts}
		\State For any $i \in [J]$, if $\wt{\mathbf{v}}_i < t$, then $\wt{\mathbf{v}}_i \gets 0$\label{line:filtering} \\
		\Comment{\textbf{Filtering step:} Removing bins below $t$}\label{line:t}
		\State Output $Q:=\left\{(c_i,\wt{\mathbf{v}}_i) \text{ for $i \in [J]$ if $\wt{\mathbf{v}}_i>0$} \right\}$ %where $c_i$'s are centers of bins.
    \label{line:output}
		\EndProcedure
	\end{algorithmic}
\end{algorithm}

First we prove that the output of Algorithm~\ref{alg:rounding-new} is differentially private. Then, we analyze the performance of Algorithm~\ref{alg:rounding-new} when $t = 0$ in terms of the worst-case utility-privacy trade-off, i.e. the case of a general input dataset where we do not impose any assumptions on its distribution. Finally, we present the utility-privacy trade-off of Algorithm~\ref{alg:rounding-new} for the special case of input data coming from a mixture of Gaussians with a positive filtering threshold. Theorems \ref{thm:general-1} and \ref{thm:general-2} give clues on how to set up width $w$ and threshold $t$ in order to achieve better utility.

\subsubsection{Differential privacy}
\begin{thm}[DP]\label{lem:DP-simple-alg}
Output of Algorithm~\ref{alg:rounding-new} is $\epsilon$-DP.
\end{thm}
\begin{proof}
Note that bins centers are picked independently of the input data.\footnote{Here we impose a mild assumption that any two neighboring datasets will be in the same axes aligned hypercube.} %with edge length $R$.}
For $J$-dimensional point counts (line~\ref{line:v} of Algorithm~\ref{alg:rounding-new}) $v$ and $\hat{v}$ corresponding to two datasets $P$ and $\hat{P}$ that differ in exactly one element, we have
$||v-\hat{v}||_1\le 2 $.
 Thus, the $\ell_1$-sensitivity is at most $2$, and we can get $\epsilon$-DP by adding $\text{Lap}(2/\epsilon)$ noise
 to each entry of the $J$ dimensional embedding (line \ref{line:laplace-noise}). Post processing feature gives that removing bins with noisy counts less than threshold $t$  (line \ref{line:filtering}) does not yield additional privacy leaks.
\end{proof}
\subsubsection{Worst-case utility-privacy trade-off ($t = 0$ case)}
We now analyze the worst case utility of Algorithm~\ref{alg:rounding-new}, i.e., the general case when we do not impose any assumptions on the distribution of the input dataset $P$.

\begin{thm}[Worst-case trade-off of Algorithm~\ref{alg:rounding-new}]\label{thm:general-1}
Suppose that dataset $P$ lies on an axes aligned hypercube of edge length $R$ in $\mathbb{R}^d$. Let $\delta>0$ be such that
$\left(\frac{R}{w}\right)^d<\frac{\epsilon n}{4\log \frac{1}{\delta}}$.
Then Algorithm~\ref{alg:rounding-new} outputs $\epsilon$-DP dataset $Q$ such that
\begin{align*}
   \sup_{x\in\mathbb{R}^d} |\KD_Q(x)-\KD_{P}(x)|\le \frac{2}{\frac{\epsilon n}{4J\log \frac{1}{\delta}} - 1} + \frac{w}{2}\sqrt{\frac{d}{e}},
\end{align*}
with probability at least $1-\delta$, where $J=\left(\frac{R}{w}\right)^d$. 
\end{thm}
\begin{proof}[Proof sketch] For a detailed proof see Section B %\ref{Appendix: data indep no filter} 
of Supplementary Material. Theorem~\ref{lem:DP-simple-alg} guarantees that the output of Algorithm~\ref{alg:rounding-new} is $\epsilon$-DP. Let $P':=\{(c_1,v_1),\ldots,(c_J,v_J)\}$, where $v$, $c$ are defined as in line~\ref{line:v} and \ref{line:c} of Algorithm~\ref{alg:rounding-new}, respectively.
We have the following sources of error.
\begin{itemize}
    \item Rounding to the bin centers. We prove that $\sup_{x\in \mathbb{R}^d}|\KD_{P'}(x) -\KD_{P}(x)|\le \frac{w\sqrt{d}}{2\sqrt{e}}$.
    \item Adding noise and removing negatively weighted bins. We prove $\sup_{x\in\mathbb{R}^d}|\KD_Q(x)-\KD_{P'}(x)| \le \frac{8J\log \frac{1}{\delta}}{\epsilon n - 4J\log \frac{1}{\delta}}$as a consequence of upper bound on $J$.  
\end{itemize}
Triangle inequality over the sources of error completes the proof.
\end{proof}

For $t=0$, Algorithm~\ref{alg:rounding-new} suffers some obvious shortcomings. If data is well spread, there will be many bins with small number of points. After Laplace noise addition, corresponding point counts would often be negative and consequently the bins would be removed. On the other hand, with probability $0.5$ empty bins would exhibit positive noisy point counts and would thus falsely be represented in the output. Both aspects hurt utility. A natural way to overcome these shortcomings is to increase the cut-off threshold. 
This approach is particularly well suited if we know that (most) non-empty bins are densely populated. 

\subsubsection{Beyond worst-case utility-privacy trade-off ($t>0$ case)}
For appropriate non-zero threshold, we present utility-privacy trade-off in the case of a general input, i.e. with no assumptions on input distribution.
\begin{defn}
For $t>0$, bins with noiseless count less (greater than or equal) than $t$ will be called $t$-light ($t$-heavy). 
\end{defn}

\begin{thm}[Beyond worst-case trade-off of Algorithm~\ref{alg:rounding-new}]\label{thm:general-2}
Suppose that input dataset $P$ lies on an axes aligned hypercube of edge length $R$ in $\mathbb{R}^d$. Assume that $\delta>0$ is such that $\left(\frac{R}{w}\right)^d \le \frac{1}{\delta}$. For  $t = \frac{8}{\epsilon}\log ({1}/{\delta})$,
let $M$ and $m$ be %the upper bound on 
the total number of $t/2$-heavy bins and 
the total number of points in $3t/2$-light bins, respectively.
Then, Algorithm~\ref{alg:rounding-new} outputs $\epsilon$-DP dataset $Q$ such that 
\begin{align*}
\sup_{x\in\mathbb{R}^{d}}    |\KD_{Q}(x)-\KD_{P}(x)|\le &\frac{\epsilon m + 8M\log \frac{1}{\delta}}{\epsilon n - \epsilon m - 4M\log \frac{1}{\delta}} +\frac{m}{n}\\
    &+ \frac{w\sqrt{d}}{2\sqrt{e}},
\end{align*}
with probability at least $1-\delta$.
\end{thm}
\begin{proof}[Proof sketch]
For detailes see Section C %\ref{Appendix: data indep with filtering} 
of Supplementary Material.
Theorem~\ref{lem:DP-simple-alg} guarantees that the output of Algorithm~\ref{alg:rounding-new} is $\epsilon$-DP. Let $P':=\{(c_1,v_1),\ldots,(c_J,v_J)\}$, where $v$, $c$ are defined as in lines ~\ref{line:v} and \ref{line:c} of Algorithm~\ref{alg:rounding-new}. We have the following sources of error.
\begin{itemize}
    \item Rounding to the bins centers. We prove that
$\label{eq:tripart1-alg-filtering}
 \sup_{x\in\mathbb{R}^{d}}|\KD_{P}(x)-\KD_{P'}(x)|\le \frac{w\sqrt{d}}{2\sqrt{e}}$.
\item With probability $1-{\delta}/{2}$, all bins that are removed in filtering step are $3t/2$-light. Thus, there are at most $m$ points that are filtered out which contributes to
$\sup_{x\in\mathbb{R}^d}|\KD_{P'}(x)-\KD_{Q}(x)|$  by at most $\frac{m}{n}$. 
\item With probability $1-{\delta}/{2}$, all bins that survive filtering step are $t/2$-heavy. Before noise addition step, number of points in these bins is at least $n-m$ and at most $n$. As there are at most $M$ such bins, the total noisy count in these bins is at least $n-m-\frac{4M}{\epsilon} \log ({\frac{1}{\delta}})$ and at most $n+\frac{4M}{\epsilon} \log ({\frac{1}{\delta}})$. This yields additional $\frac{\epsilon m + 8M\log (1/\delta)}{\epsilon n - \epsilon m - 4M\log (1/\delta)}$ term for the upper bound. 
\end{itemize}
Union bound and triangle inequality
over the sources of error complete the proof.
\end{proof}

The threshold $t$ in Theorem \ref{thm:general-2} depends on privacy level $\epsilon$, and so both $M$ and $m$ depend on $\epsilon$. With no assumptions on distribution of $P$, it is not possible to provide meaningful bounds on $M$ and $m$. We next study a special case. 

\subsubsection{Mixture of Gaussians input data}

We analyze performance of Algorithm \ref{alg:rounding-new} for the special case of input data from multivariate Gaussian distribution. More precisely, we consider dataset $P$ of $n$ records in $\mathbb{R}^d$ with Gaussian distribution $\mathcal{N}(\mathbf{c},\sigma^2I)$, $\mathbf{c}\in\mathbb{R}^d$ i.e., from density $f(X=x)=\frac{1}{(2\pi\sigma^2)^{d/2}}e^{-\frac{||x-\mathbf{c}||_2^2}{2\sigma^2}}$. This straightforwardly generalizes to a mixture of multivariate Gaussians.

\begin{thm}[Gaussian trade-off using Algorithm~\ref{alg:rounding-new}]\label{thm one gaussian tradeoff}
Suppose that input dataset $P$ lies on an axes aligned hypercube of edge length $R$ in $\mathbb{R}^d$. If $n\geq \left(\frac{w}{\sigma \sqrt{2\pi}}\right)^d$, for $\delta>0$ such  that $\frac{n}{(\log n)^{d/2}}\geq 16\cdot \log \frac{1}{\delta}\cdot (\frac{12\sigma}{w})^2$ and threshold $t = \frac{8}{\epsilon}\log ({1}/{\delta})$, Algorithm~\ref{alg:rounding-new} outputs $\epsilon$-DP dataset $Q$ s.t.
\begin{align*}
\sup_{x\in\mathbb{R}^d}    |\KD_{Q}(x)-\KD_{P}(x)| &\le \frac{8(\log \frac{1}{\delta})^{1/3}\cdot e^{-\frac{d}{3}(\log \frac{w}{\sigma \sqrt{2 \pi}}-2)}}{(\epsilon n)^{1/3}}\\
    &+\frac{16 \log \frac{1}{\delta}\cdot(\frac{12\sigma}{w})^d (\log n)^{d/2}}{\epsilon n}\\
    &+ \frac{w\sqrt{d}}{2\sqrt{e}},
\end{align*}
with probability at least $1-\delta$.
\end{thm}

\begin{proof}[Proof sketch]
For a detailed proof see Section D %~\ref{Appendix Gausssian} 
of Supplementary Material. This is a special case of Theorem \ref{thm:general-2}. The Gaussian distribution assumption enables us to provide upper bounds on the number of $t/2$-heavy bins $m$, and total number of points in $3t/2$-light bins $M$, for $t=\frac{8}{\epsilon}\log(1/\delta)$. Loosely speaking, in this case majority of points lives within densly populated areas, and so there are neither too many points in the light bins nor the number of heavy bins is too large.
\end{proof}
In Section \ref{sect:experiment_theoretical_scaling} we compare empirical MMD to the bound $(O(\epsilon n)^{-1/3})$ from Theorem \ref{thm one gaussian tradeoff}. Results suggest that lower bound is smaller which is in line with results from \cite{Duchi} where the gap is of order $O(n^{-1/2})$. Their setting is different as they study local differential privacy in the context of distribution estimation (not synthetic data release). 
We also highlight recent work of  \cite{DBLP:journals/corr/abs-1805-00216} on DP learning parameters of multivariate Gaussian distribution. Their approach learns parameters and is thus not directly comparable as we output a dataset.

\subsection{Data dependent algorithm}

In  the general case, data independent approaches suffer from the curse of dimensionality. The reason is that as opposed to traditional applications of hash functions, we need to keep track of empty bins in order to treat them similarly to non empty ones, as a bin that is empty with respect to $P$ is not necessarily empty with respect to a neighbouring dataset. In high dimensions, this makes the data independent binning impractical, as there are typically many empty bins. Moreover, in densely populated areas finer grid would incur smaller error due to the rounding to the centers, and thus yield a higher utility. %Both aspects motivate data dependent binning approaches. 

We propose differentially private algorithm based on \textit{adaptive binning} i.e., recursive partitioning of the space. Before we proceed, we introduce notion of a \textit{decision tree}. We assume arbitrary but fixed enumeration of $d$ dimensions denoted by $i$ where $i\in[d]$. The root of the tree is characterized by the initial dataset $P$, the center $c$ and the the edge $R$ of the smallest axis aligned cube that contains whole $P$, and axes $0$ to split along. If decision at the root is to proceed with recursion (we discuss decision making below), we proceed as follows. Initial recursion splits the dataset along axes $0$ and divides the corresponding edge of the cube in two equal $R/2$ parts which results in the creation of two \textit{children} nodes. Each of them is characterized by a fraction of the dataset that ended up in the corresponding part, center and the radius of the new cube that contains that fraction of the dataset, and new axes to cut along. The new axes to cut along is always previous axes $+1$ i.e., the next axes in ordering $[d]$. The recursion proceeds on the newly created nodes, subject to a positive decision on whether to recurse further. Nodes on which recursion does not proceed do not have any children and represent \textit{leaves}.

Note that due to the data dependent aspect, we need to ensure that adaptive binning is differentially private. Thus, each decision on whether to recurse or not has to be based on noisy point counts. By composition property of DP, if there are $l$ levels of data dependent decisions in the decision tree, one needs to guarantee ${\epsilon'}{/l}$-DP for each recursion, so that once binning algorithm terminates we achieve $\epsilon'$-DP. Note that the output of adaptive binning is partitioning of the space, not the synthetic dataset, and thus once the binning is done we are in the setting from the beginning of Algorithm \ref{alg:rounding-new}. That is, in order to release differentially private synthetic dataset, we have to obtain noisy versions of the point counts bin by bin and output bins that pass certain threshold. By ensuring a $\epsilon''$-DP for this part of the procedure, one will achieve $(\epsilon'+\epsilon'')$-DP guarantee for the whole algorithm by the composition property of DP. Before we formally introduce adaptive binning, we discuss some desirable settings.
\paragraph{Avoiding large bins} In high dimensions, we will frequently observe large bins where recursion stops due to small noisy point counts. This would however yield large error due to rounding to the center, since the bin has large edge lengths. Thus it would be beneficial to have the algorithm run bin splitting \textit{independently} of data for a few rounds, e.g., until it reaches a maximum edge length for all bins below some threshold $s_{1}$. After that, the algorithm would run in data dependent regime. Note that the depth of data independent part in this setting is $h=d\log_{2}({R}/{s_{1}})$, and up until that level there are no privacy leaks. 
\paragraph{Avoiding decision trees with large depth}
Privacy cost of adaptive binning is determined by the number of data dependent levels in decision tree. Thus, even if noisy point counts are large, it might be beneficial to stop the recursion once the number of data dependent levels passes certain threshold. Equivalently, we stop the recursion over a bin when its maximal edge length is below certain threshold $s_{2}$, regardless of the value of the noisy point count. In this setting, the maximum tree depth is $h'=d\log_{2}({R}/{s_{2}})$, and having in mind above discussion, number of data dependent levels is at most $h'-h=\log_{2}({s_{1}}/{s_{2}})$.

Algorithm \ref{alg:adaptive-binning}  formally describes our adaptive binning. More precisely, it identifies the root of the tree as discussed above, and passes it to Algorithm \ref{alg:recursive-binning}  (discussed below). Lines \ref{line: min along dim} - \ref{line: max along dim} identify the boundaries of the dataset along each of $d$ dimensions, line \ref{line: center along dim} identifies the center of the cube, and line \ref{line: edge cube} the edge of the cube that contains the whole dataset. Axes for next split is set to $0$, and the root node is passed to Algorithm \ref{alg:recursive-binning} .
\begin{algorithm}
	\caption{Adaptive binning}  
	\label{alg:adaptive-binning} 
	\begin{algorithmic}[1]
  \Procedure{$\textsc{Adaptive-Binning}(P,\epsilon',\tau)$}{} 
		%		\State {} \Comment{When the query comes we just locate the query in each round of hashing} 
		\\
		\Comment {dataset $P$, DP budget $\epsilon'$,cut off level $\tau$}
 		\State For any $i \in [d]$, $\text{low}_i \gets \min_{p \in P}(p_i)$ %\Comment{$p_i$ is the $i$'th coordinate of $p$}
        \label{line: min along dim}
		\State For any $i \in [d]$, $\text{high}_i \gets \max_{p \in P}(p_i)$
        \label{line: max along dim}
		\State For any $i \in [d]$, $\text{c}_i \gets \frac{\text{low}_i+\text{high}_i}{2}$ 
  \label{line: center along dim}
  %\Comment{$c$ is the center of an axis aligned cube with the smallest edge length that includes the whole dataset}
 		\State $w \gets \max_i(\text{high}_i-\text{low}_i)$
        %\Comment{$w$ \textcolor{red}{here coincides with $R$ i.e.} is the smallest value such that an axes aligned cube with edge length $w$ can include the whole dataset} 
        \label{line: edge cube}
        \State $\text{curr-axis} \gets 0$ \Comment{\text{ axis the bin will be cut along}} 
        \State $\text{node} \gets \textsc{Node}(P, c, w, \text{curr-axis})$
        \State $\textsc{Recursive-Binning}(\text{node},\epsilon',\tau)$
		\EndProcedure
	\end{algorithmic}
\end{algorithm}
Finally, Algorithm \ref{alg:recursive-binning} implements differentially private recursive binning. It takes as an input a node, total privacy budget (until the recursion stops), threshold for the noisy point counts, and the maximum and minimum allowed edge lengths of final bins. The output of the algorithm is the set of bins. According to previous discussion, the algorithm will recurse if either noisy point count is larger than the threshold, or the bin's largest edge is too large (line \ref{line:condition}), with the exception of that if the largest edge is too small (line \ref{line:condition_bin_too_small}), recursion stops regardless of the value of the noisy point count. If recursion proceeds, the current node is split in the two and recursion proceeds on each of them.

%\begin{algorithm}[H]
\begin{algorithm}
	\caption{Recursive binning}  
	\label{alg:recursive-binning} 
	\begin{algorithmic}[1]
		
		\Procedure{$\textsc{Recursive-Binning}(\text{node},\epsilon',\tau, s_{1}, s_{2})$}{}\\
		\Comment{Node to recurse on, DP budget $\epsilon'$, threshold $\tau$}\\
        \Comment{max and min edge length for final bins $s_{1}$ and $s_{2}$}
		\State $P,c,w,\text{curr-axis} \gets \text{node}$
		\If {$|P| \le \tau + \textsc{Lap}(2(h'-h)/\epsilon')$ and the bin's largest edge length $\le s_1$}
            \Return \label{line:condition}\\  
           \Comment{$h'-h= d\log_2(s_{1}/s_2)$ is the max depth of data dependent part of the tree}
        %\Comment{$s_1$ is the maximum edge length for final bins}		\\
            %\Comment{$s_2$ is the minimum edge length for final bins}	
		\EndIf
		\If {the bin's largest edge length $< s_2$}
		\Return\ \label{line:condition_bin_too_small}
%  \Comment{The recursion stops when all edge lengths of the bin are below $s_2$}
		\EndIf
		
		\State $P_{\ell}, P_r \gets \emptyset,\emptyset$
% 		\State $P_{r} \gets \emptyset$
		
		\For{ $p \in P$}
		\If {$p_{curr-axis} \le c_{\text{curr-axis}}$}
		$P_{\ell}\gets P_{\ell} \cup \{ p \}$
		\Else
		 \ $ P_{r}\gets P_{r} \cup \{ p \}$
		\EndIf 
		\EndFor
		\State $c_\ell , c_r \gets c , c$
		\State $c_{\ell,{curr-axis}} \gets c_{curr-axis} - w/4$ %\Comment{curr-axis'th coordinate of $c_\ell$ gets updated}
% 		\State $c_r \gets c$
		\State $c_{r,{curr-axis}} \gets c_{curr-axis} + w/4$ %\Comment{curr-axis'th coordinate of $c_r$ gets updated}
  \Comment{coordinates of newly formed bins centers}
		\If {$\text{curr-axis} = d-1$} 
		 $w \gets w/2$
		\EndIf
% 		\If {$P_\ell \ne \emptyset$}
		    \State $\text{node.left} \gets \textsc{Node}(P_\ell,c_\ell,w,(\text{curr-axis}+1)\% d)$
		    \State $\textsc{Recursive-Binning}(\text{node.left},\epsilon',\tau)$
% 		\EndIf
% 		\If {$P_r \ne \emptyset$}
		    \State $\text{node.right} \gets \textsc{Node}(P_r,c_r,w,(\text{curr-axis}+1)\% d)$
		    \State $\textsc{Recursive-Binning}(\text{node.right},\epsilon',\tau)$
% 		\EndIf
		\EndProcedure
	\end{algorithmic}
\end{algorithm}

\begin{thm} \label{thm adaptive binning DP}
For any dataset $P$, $\epsilon'>0$, $\tau > 0$, $s_{1},s_{2}>0$, Algorithm \ref{alg:adaptive-binning} returns a tree such that the set of bins determined by its leaves is $\epsilon'$-differentially private.
\end{thm}
\begin{proof}
Data dependent part of the decision tree for this algorithm has maximum depth $h'-h = d\log_2(s_{1}/s_2)$. Each neighboring dataset affects only one root to leaf path, and thus it is enough to consider maximal privacy loss incurred along a single path to the leaf. Also, since the $\ell_1$ sensitivity of the point counts is $1$, it suffices to add $\textsc{Lap}(2(h'-h)/\epsilon')$ noise to each decision condition in order to guarantee $\epsilon'$-differential privacy for the entire path. Two neighboring datasets differ in the point counts in at most two paths.
\end{proof}
\begin{thm}[DP of Data Dependent algorithm] \label{thm data dependant algo DP}
    For a dataset $P$ and $\epsilon'>0$, %$\tau>0$, 
    let 
    $c_{1},\dots,c_{J}\in \mathbb{R}^{d}$ denote centers of bins corresponding to the leaves of tree output by Algorithm \ref{alg:adaptive-binning}, and $\mathbf{v} \in \mathbb{R}^J$ be the vector of corresponding point counts in each bin. If $c_{1},\dots,c_{k}$ and $\mathbf{v}$ are passed to line \ref{line:laplace-noise} of Algorithm \ref{alg:rounding-new} with $\epsilon''>0$, then the final output of Algorithm \ref{alg:rounding-new} is $\epsilon'+\epsilon''$-DP dataset $Q$.
\end{thm}
\begin{proof}
    Consequence of composition property of DP.
\end{proof}

\subsubsection{Implicit sampling of empty bins}

As already discussed, number of empty bins grows exponentially in dimension and this represents a challenge as we need to treat empty bins in a same manner as non empty ones for the purpose of achieving DP. Thus, efficient implementation of our algorithms would avoid storing all bins and iterating through them for the purpose of noise addition and filtering. We utilize idea exploited in \cite{cormode2011differentially}, to implicitly implement the noise addition and filtering on empty bins. This benefits both data independent Algorithm \ref{alg:rounding-new} (see Section E.1 %\ref{appendix secti implicit data indep}
of Supplementary Material) and data dependent Algorithm \ref{alg:adaptive-binning} (discussed below).

\begin{lem} \label{lem:sampling explicit equivalent implicit}
    Explicit implementation of noise addition and filtering on empty bins as per Algorithm \ref{alg:explicit-empty-bins} is equivalent to the implicit implementation provided in Algorithm \ref{alg:implicit-empty-bins} .
\end{lem}

See Section F %\ref{appendix sect implicit}
of Supplementary Material for  proof. In Algorithm \ref{alg:implicit-empty-bins}, $\textsc{ConditionalLap}(2/\epsilon,t)$ denotes a random variable with Laplace distribution conditioned on being greater than or equal to $t$.

\begin{algorithm}
	\caption{Explicit implementation}  
	\label{alg:explicit-empty-bins} 
	\begin{algorithmic}[1]
		\For {any empty bin}
        \State $\eta \gets \textsc{Lap}(2/\epsilon)$
        \If{$\textsc{PointCount}(\text{bin}) + \eta = 0 + \eta \ge t$}
        \State add this bin center with weight $\eta$ to the output %synthetic dataset
        \EndIf    
        \EndFor
    \end{algorithmic}
\end{algorithm}

\begin{algorithm}
	\caption{Implicit implementation}  
	\label{alg:implicit-empty-bins} 
	\begin{algorithmic}[1]
		\State $K \gets$ the number of empty bins
		\State $p \gets \Pr[\textsc{Lap}(2/\epsilon) \ge t]$
		\State $m \gets \textsc{Binom}(K,p)$
		\State sample $m$ empty bins out of $K$ empty bins without replacement
		\For {each sampled empty bin}
		\State $\eta \gets $ a sample of $\textsc{ConditionalLap}(2/\epsilon,t)$  
		\State add this bin center with weight $\eta$ to the output 
		\EndFor
    \end{algorithmic}
\end{algorithm}

\subsubsection{Implicit sampling in the data dependent Algorithm \ref{alg:adaptive-binning}} 

The union of the set of empty and non empty bins given as output of Algorithm \ref{alg:adaptive-binning} coincides with the set of leaves in the tree of recursion decisions. 
The question is, whether from the information on the number of data independent levels $h$, and paths from the root to the leaves corresponding to non empty bins, one can recover the remaining set of leaves i.e., those corresponding to the empty bins. In this section, we prove that this is possible for a large enough cut off threshold $\tau$ (see line~\ref{line:condition} in Algorithm~\ref{alg:recursive-binning}) which guarantees that with high probability the algorithm does not partition empty bins further in the data dependent part of the algorithm. See Figure \ref{fig:decision-tree} for an illustration of a decision tree, and Section F %\ref{appendix sect implicit} 
of Supplementary Material for proof of Lemma \ref{lem:empty bin not split}.

\begin{lem} \label{lem:empty bin not split}
Let $h'$ and $h$ denote total depth and the depth of data independent part of the tree, respectively. If the threshold $\tau$ in line~\ref{line:condition} of Algorithm~\ref{alg:recursive-binning} is set to be greater than 
$$ \frac{2(h'-h)}{\epsilon'}\log \left(\frac{1}{\delta}\cdot\left(2^h + n(h'-h)\right)\right),$$
then with  probability $1-\delta$ the adaptive binning will not divide any empty bin.
\end{lem}

\begin{thm}
If we pick threshold $\tau$ in line~\ref{line:condition} of Algorithm~\ref{alg:recursive-binning} as per Lemma \ref{lem:empty bin not split}, then storing information on the number of data independent levels $h$ and paths to leaf nodes that represent non empty bins, enables implicit sampling of empty bins as per Algorithm \ref{alg:implicit-empty-bins}.
\end{thm}

\begin{proof}
    Lemma \ref{lem:empty bin not split} guarantees that no empty bin is recursed on. Thus, each parent node has at most one child corresponding to an empty bin. In particular, for each two non-empty bins (black nodes in Figure ~\ref{fig:decision-tree}
    ) we can identify their common ancestor and the number of empty bins (gray nodes in Figure \ref{fig:decision-tree}) between them. Thus, set of empty bins can be recovered from the encoding of non empty bins. See Section F %\ref{appendix sect implicit}
    of Supplementary Material for details.
\end{proof}

\begin{figure}
\centering
\scalebox{0.5}{
\begin{forest}
for tree={
    grow=south,
    circle, draw, minimum size=3ex, inner sep=0.75pt,
    s sep=2mm
        }
[
    [
        [, edge=dashed,fill=gray
        [,no edge, draw=none]
        [,no edge, draw=none]
        ]
        [ 
            [ 
            [
            [,fill=black[,no edge, draw=none]
        [,no edge, draw=none]]
            [, edge=dashed,fill=gray
            [,no edge, draw=none
            ]
        [,no edge, draw=none]]
            ]
            [, edge=dashed,fill=gray
            [,no edge, draw=none
            ]
        [,no edge, draw=none]]
            ]
            [, edge=dashed,fill=gray
            [,no edge, draw=none]
        [,no edge, draw=none]
        ]
        ]
    ]
    [
        [
            [, edge=dashed,fill=gray
            [,no edge, draw=none]
        [,no edge, draw=none]
        ]
            [,fill=black
            [,no edge, draw=none]
        [,no edge, draw=none]] 
        ]
        [, edge=dashed,fill=gray
        [,no edge, draw=none]
        [,no edge, draw=none]
        ]
    ]
]
\end{forest}}
\vspace{-0.5cm}
\caption{Tree with $h = 2$ and $h' = 5$. Black nodes are non-empty bins, gray nodes are empty bins we need to sample.}\label{fig:decision-tree}
\end{figure}

\section{Experiments}\label{sect experiments}
We provide experimental results on both of our proposed algorithms in various settings.
\paragraph{Privacy-utility trade-off}
\begin{figure*}[!ht]
    \centering
    \begin{subfigure}{0.33\textwidth}
    \centering
    \includegraphics[width=\textwidth]{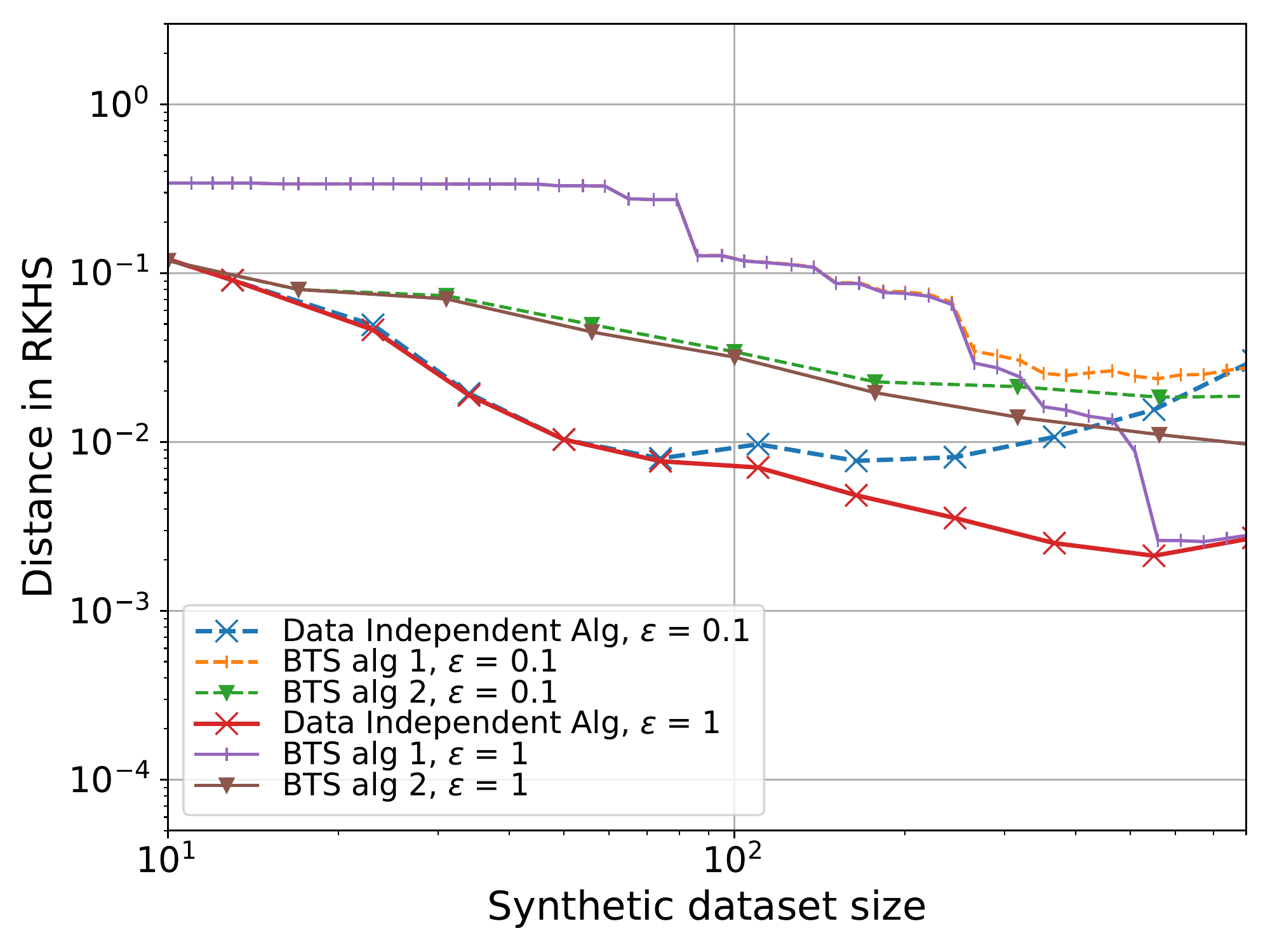}
    \caption{}
    \label{fig:DI-D2}
    \end{subfigure}
     \hfill
    \begin{subfigure}{0.33\textwidth}
    \centering
    \includegraphics[width=\textwidth]{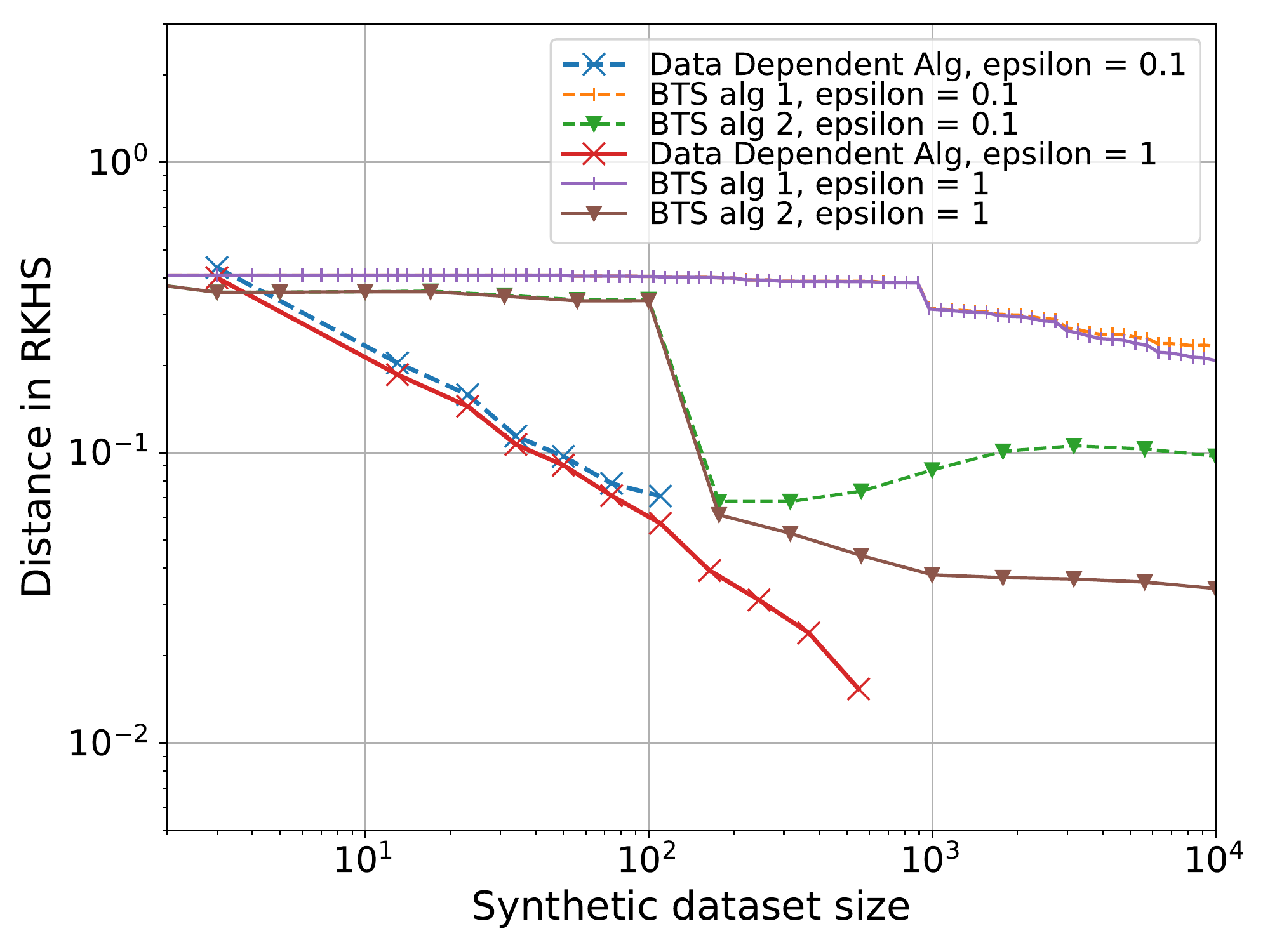}
    \caption{}
    \label{fig:DD-D5}
    \end{subfigure}
     \hfill
    \begin{subfigure}{0.33\textwidth}
    \centering
    \includegraphics[width=\textwidth]{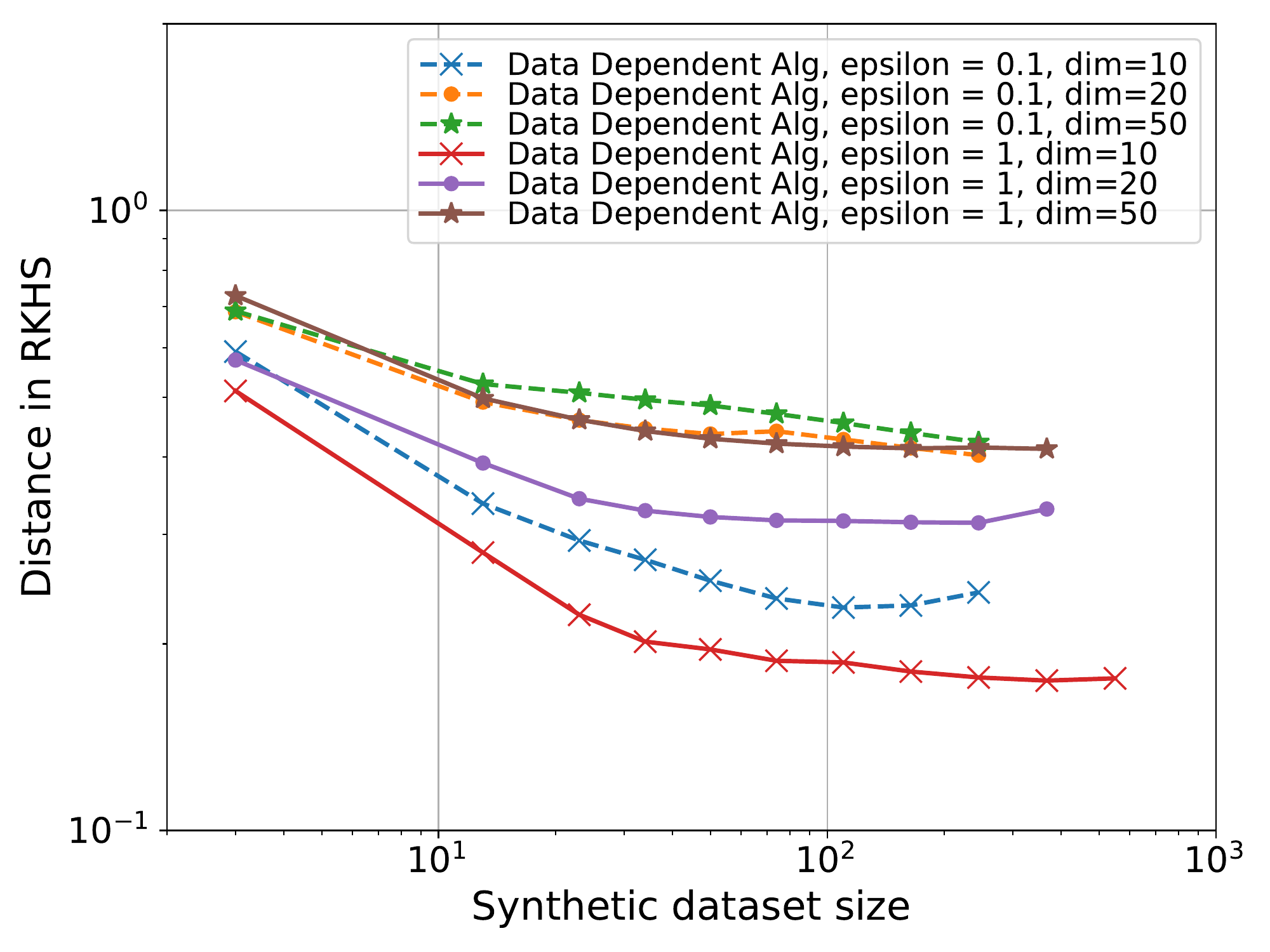}
    \caption{}
    \label{fig:DD-various}
    \end{subfigure}
    \caption{ (Left) Algorithm~\ref{alg:rounding-new} vs \cite{DBLP:conf/icml/BalogTS18} in dimension 2. (Middle) Our Data Dependent Algorithm vs \cite{DBLP:conf/icml/BalogTS18} in dimension 5.
    (Right) Data Dependent Algorithm in high dimensions.}
    \label{fig:d2vsBalog}
\end{figure*}

First, we show improvements using our data independent Algorithm~\ref{alg:rounding-new} over both algorithms of \cite{DBLP:conf/icml/BalogTS18} on a mixture of Gaussians dataset in dimension 2. The comparison is presented in Figure~\ref{fig:DI-D2}, where a considerable improvement is achieved, in all regimes of privacy budget $\epsilon$ and dataset size. For a 5-dimensional mixture of Gaussians dataset, although our data independent algorithm is unable to outperform the algorithms of \cite{DBLP:conf/icml/BalogTS18}, %\Navid{@EK @VP check plz if you like},
our data dependent algorithm overcomes the curse of dimensionality and it achieves lower error rates using smaller number of synthetic data points. This improvement for all regimes of privacy levels is presented in Figure~\ref{fig:DD-D5}. For both settings, we use datasets utilized in \cite{DBLP:conf/icml/BalogTS18} (see Section G.1 %\ref{appendix_balog_exp_setting}
of Supplementary Material for details). Figure~\ref{fig:DD-various} confirms that our data dependent algorithm achieves low errors for high dimensions and various privacy regimes, and thus overcomes curse of dimensionality.

\paragraph{Tightness of bound from Theorem \ref{thm one gaussian tradeoff}}\label{sect:experiment_theoretical_scaling}

Theorem \ref{thm one gaussian tradeoff} provides the upper bound on utility loss of Algorithm \ref{alg:rounding-new} for samples from multivariate Gaussian. When bins are such that the error $\frac{w\sqrt{d}}{2\sqrt{e}}$ arising from rounding to centers is negligible, for $\epsilon \gg\frac{1}{n}$ the bound scales as $O\left((\epsilon n)^{-1/3}\right)$. Figure \ref{fig:various_dims_vs_scaling} shows empirical MMD for various privacy levels vs. $O\left((\epsilon n)^{-1/3}\right)$ benchmark and suggests that there might be a gap and our bound can be improved. See Section \ref{sect: future work} for further discussion.

\begin{figure}[ht!]
    \centering   \includegraphics[width=0.33\textwidth]{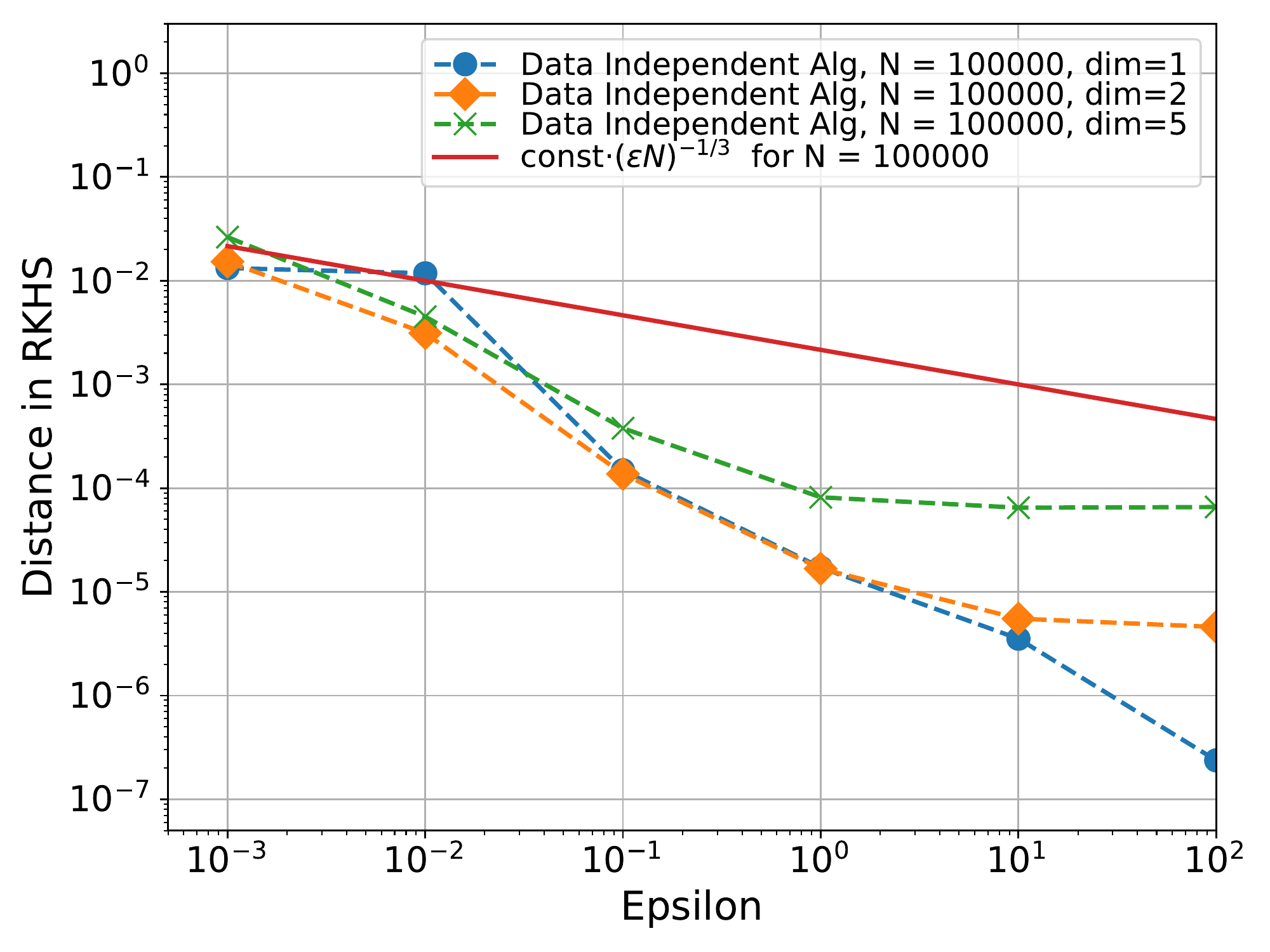}
    \caption{Utility loss of Algorithm \ref{alg:rounding-new} for various privacy budgets, and input of $100,000$ samples from standard Gaussian for various dimensions. It  stabilizes regardless of privacy budget due the error of rounding to centers.}
    \label{fig:various_dims_vs_scaling}
\end{figure}

\paragraph{Classification task on real data}

We evaluate performance of our data dependent algorithm on downstream binary classification task on a real tabular dataset. We rely on credit card fraud detection dataset used in \cite{DBLP:conf/aistats/HarderAP21} and compare our data dependent algorithm with their DP-MERF. We follow their experimental setup to train $12$ classifiers on synthetic data and evaluate their performance on original data. See Section G.4 %\ref{appendix:binary_classification}
of Supplementary Material for more detailed setup and results. Our algorithm does not outperform DP-MERF in terms of ROC values which is not surprising as we do not rely on deep generative models. However, our performance degrades slower as privacy increases (Figure \ref{fig:ROC_degradation}).
\begin{figure}[ht!]
    \centering   \includegraphics[width=0.33\textwidth]{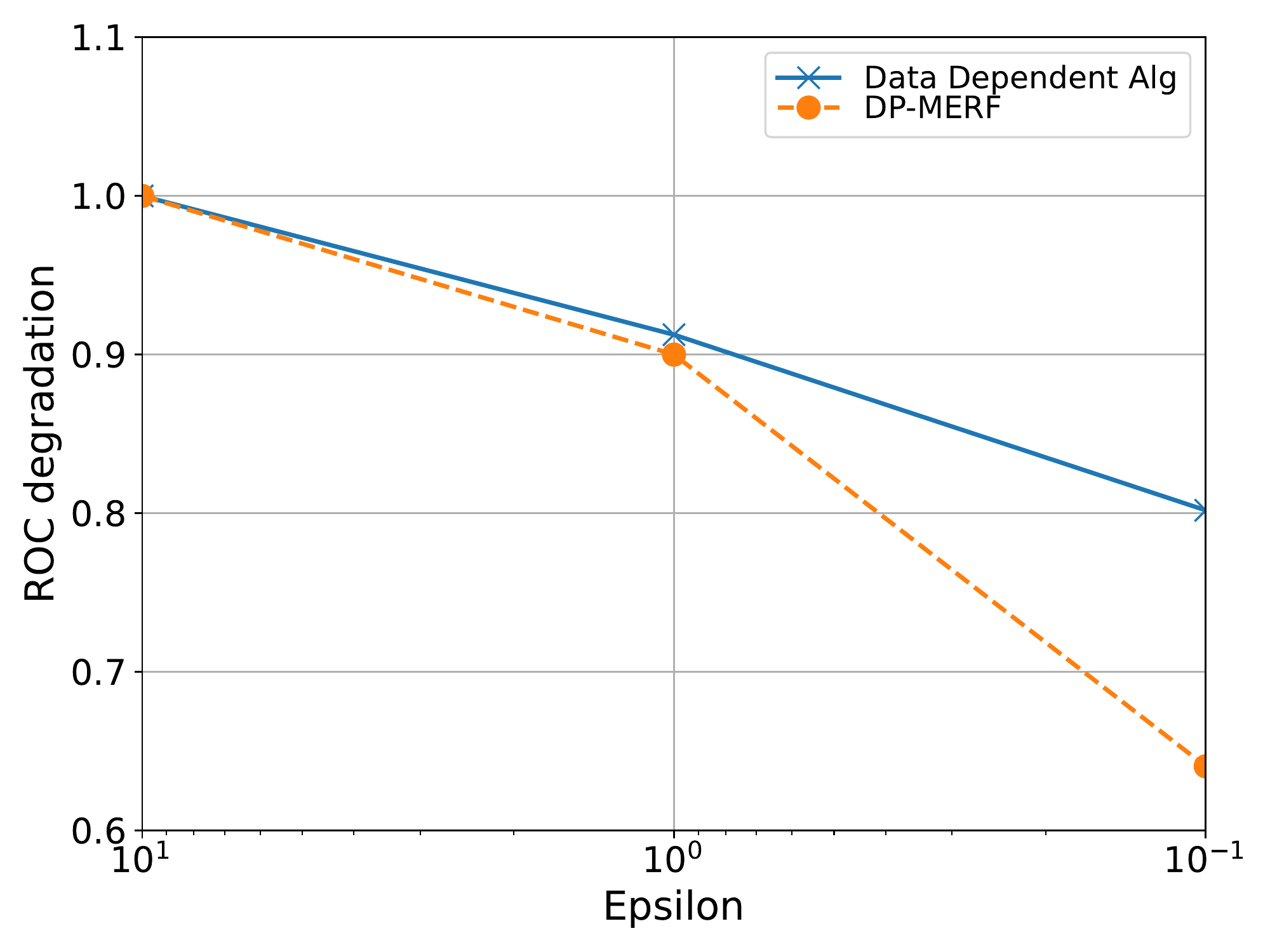}
    \caption{Comparison of our Data Dependent Algorithm and DP-MERF on downstream classification task. ROC degradation is represented as a ratio of ROC corresponding to a specified $\epsilon$ budget and ROC for $\epsilon=10$. }
    \label{fig:ROC_degradation}
\end{figure}

\paragraph{Optimal size of bins}

We explore what combination of widths and weights of bins minimizes MMD between synthetic data and the sample (regime with no privacy). For $1$-dimensional standard Gaussian input, we consider synthetic data given by a mixture of uniforms. More precisely, synthetic data consists of
centers of $2k+1$ bins of a given width, symmetrically arranged around the mean. Figure \ref{fig:mmd_kl} shows KL and MMD between standard Gaussian and mixture of uniforms \citep{https://doi.org/10.1002/sta4.329}, as well as sample MMD for a Gaussian sample of size $100,000$. For appropriately set width, sample MMD is getting small which indicates that binning approach has potential to yield good utilities. For more details and derivation of optimal widths and weights
see Section H %\ref{appendix:proxy_mixture_uniforms}
of Supplementary Material.

For additional experiments, see Sections G.2 %\ref{appendix: dep on size} 
and G.3 %\ref{appendix: dep on variance}
of Supplementary Material.

\begin{figure}[ht!]
    \centering   \includegraphics[width=0.33\textwidth]{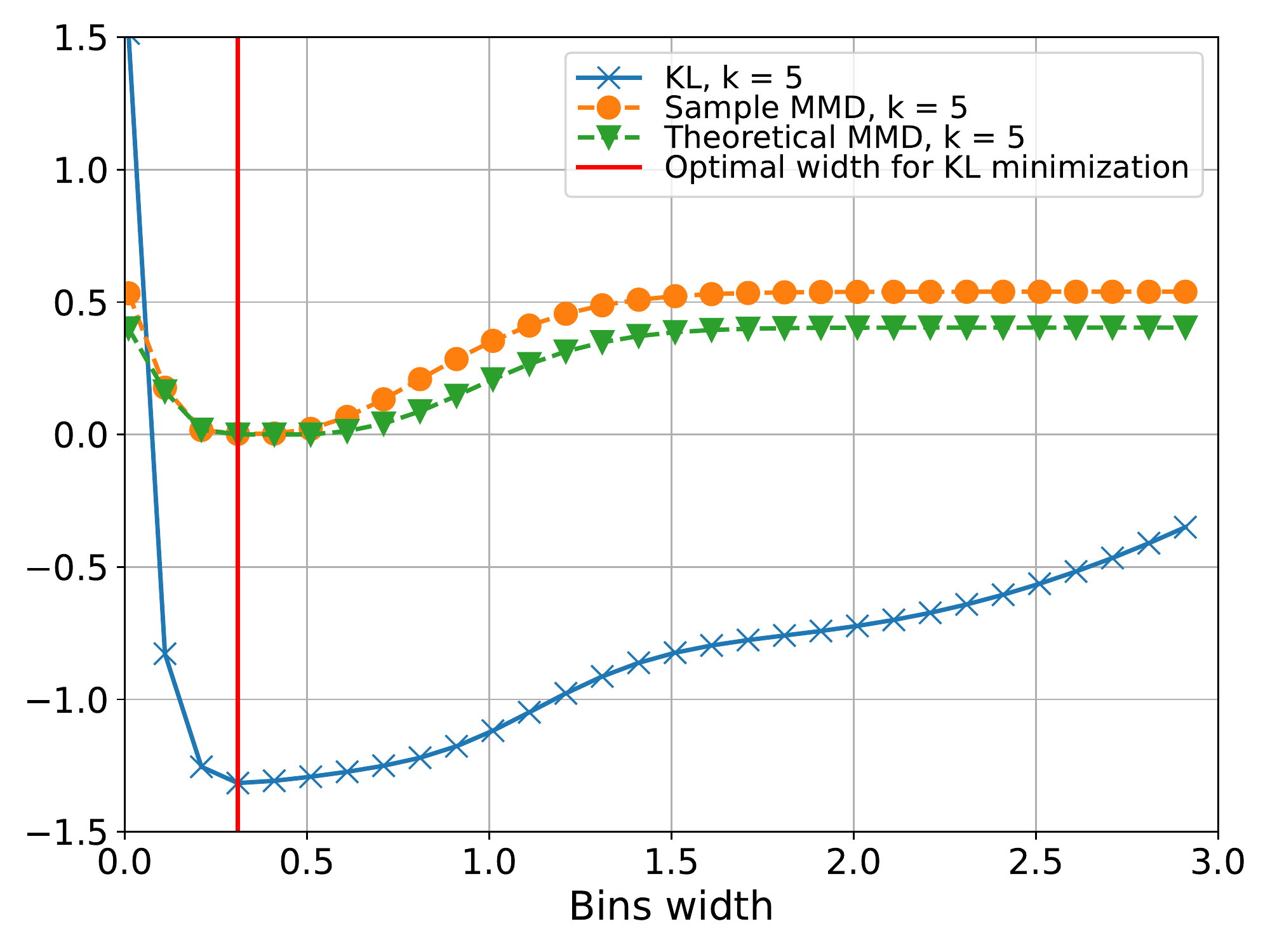}
    \caption{$1$-dimensional standard Gaussian vs $11$-mixture of uniforms and various widths.}
    \label{fig:mmd_kl}
\end{figure}

\section{Conclusions and Future work}\label{sect: future work}
We proposed an interpretable and efficient algorithms for differentially private synthetic data generation, where data takes multidimensional numerical form. We provide theoretical bounds in various settings and empirically validate our algorithm performance in terms of both simulated and real-world datasets.

Our algorithms can be extended to tabular data characterized by mixed numerical and categorical columns, as each category within a column corresponds 
to a bin along that dimension. However, this setting does not benefit from our theoretical guarantees on the utility (e.g. Theorem \ref{thm:general-1}) which we measure in terms of kernel density metrics defined on $\mathbb{R}^{d}$. Moreover, our data dependent approach does not handle the curse of dimensionality problem as the total number of bins grows exponentially with the number of categorical columns. This is due to the fact that in mixed numerical and categorical settings there is no straightforward way to merge or split categories depending on point counts, as it is the case with bins in $\mathbb{R}^{d}$. It is an interesting future direction to combine our methods for numerical columns with a different method along categorical columns in order to achieve scalable solution in such setting.

Another interesting direction is to resolve the question of the theoretical utility-privacy guarantees for our data dependent algorithm. This is challenging in the data dependent setting as it is not easy to provide meaningful bounds on the total number of bins.
Furthermore, we conjecture that one can improve upon our bounds using more sophisticated LSH functions, e.g., projection based LSHs such as \cite{DBLP:conf/compgeom/DatarIIM04,DBLP:conf/focs/AndoniI06}. This direction needs novel techniques in order to handle empty bins.
For the special case of Gaussian multivariate distribution, we leave it as future work to compare our algorithms to existing results on DP, e.g., \cite{DBLP:journals/corr/abs-1805-00216} who provide DP learning of distribution parameters. Lower bounds for our space-partitioning approach also remain open, in particular whether our bound from Theorem \ref{thm one gaussian tradeoff} can be improved to be of the form $O(n^{-1/2})$ as suggested by results of \cite{Duchi}.

%\begin{contributions} % will be removed in pdf for initial submission 
					  % (without ‘accepted’ option in \documentclass)
                      % so you can already fill it to test with the
                      % ‘accepted’ class option
%    Briefly list author contributions. 
%    This is a nice way of making clear who did what and to give proper credit.
%    This section is optional.

%    H.~Q.~Bovik conceived the idea and wrote the paper.
%    Coauthor One created the code.
%    Coauthor Two created the figures.
%\end{contributions}

%\begin{acknowledgements} % will be removed in pdf %for initial submission,
						 % (without ‘accepted’ option in \documentclass)
                         % so you can already fill it to test with the
                         % ‘accepted’ class option
%    Briefly acknowledge people and organizations here.

%    \emph{All} acknowledgements go in this section.
%\end{acknowledgements}

\paragraph{Disclaimer}
This paper was prepared for informational purposes by the Artificial Intelligence
Research group of JPMorgan Chase \& Co and its affiliates (“J.P. Morgan”), and is not a product
of the Research Department of J.P. Morgan. J.P. Morgan makes no representation and warranty
whatsoever and disclaims all liability, for the completeness, accuracy or reliability of the information
contained herein. This document is not intended as investment research or investment advice, or a
recommendation, offer or solicitation for the purchase or sale of any security, financial instrument,
financial product or service, or to be used in any way for evaluating the merits of participating in
any transaction, and shall not constitute a solicitation under any jurisdiction or to any person, if such
solicitation under such jurisdiction or to such person would be unlawful.

% References
\bibliography{kreacic_397}

\newpage
\onecolumn
\appendix

\section{Some auxiliary results}

\begin{rem}\label{fact:tailbound4laplace}
     For $X\sim \textsc{Lap}(2/\epsilon)$, we have
     \begin{align*}
         \Pr[X \ge \alpha] = \frac{1}{2}e^{-\frac{\alpha}{2/\epsilon}}.
     \end{align*}
\end{rem}

\begin{rem} \label{rem laplace tail log}
    As a consequence of Remark \ref{fact:tailbound4laplace}, for $X\sim \textsc{Lap}(2/\epsilon)$, we have
     \begin{align*}
         \Pr\left[|X| \ge \frac{4C\log n}{\epsilon}\right] = n^{-2C}.
     \end{align*}
\end{rem}

\begin{rem}
     Let $X\sim \textsc{Lap}(2/\epsilon)$, then we have
     \begin{align*}
         \mathbf{E}[|X|] = \frac{1}{\epsilon}.
     \end{align*}
\end{rem}

\begin{rem}[\cite{Laurent2000AdaptiveEO}]
     Let $Y:=\sum_{k=1}^{d} Z_k^2$, where $Z_k\sim \mathcal{N}(0,1)$ are i.i.d. random variables. Then
     \begin{align*}
         \Pr\left[Y \ge d+2\sqrt{dx}+2x \right] \le e^{-x}, \forall x>0
     \end{align*}
\end{rem}
\begin{cor}
Let $Y:=\sum_{k=1}^{d} Z_k^2$, where $Z_k\sim \mathcal{N}(0,\sigma^2)$ are i.i.d. random variables. Then
     \begin{align*}
         \Pr\left[Y/\sigma^2 \ge  2d+3x \right] \le \Pr\left[Y/\sigma^2 \ge  d+2\sqrt{dx}+2x \right] \le e^{-x}, \forall x>0
     \end{align*}
\end{cor}
\begin{cor}\label{cor:tailboundschisquared}
For any $X\sim \mathcal{N}(\mathbf{c},\sigma^2I)$, we have
\begin{align*}
    \Pr\left[||X-\mathbf{c}||_2 \ge \sigma \sqrt{2d+3x}\right]\le e^{-x} 
\end{align*}
\end{cor}

\begin{lem}[Chernoff bounds]\label{lem:chernoff}
Let $X_1,X_2,\ldots,X_n$ be independent binary random variables. Define $Y:=\sum_{i=1}^{n}X_i$ and $\mu := \mathbb{E}[Y]$. Then, for any %$\Delta>0$ and 
$\delta>0$
and 
\begin{align*}
    \Pr[|Y-\mu|>\delta\mu] \le 2 \exp(-\delta^2\mu/4).
\end{align*}
\end{lem}

\section{Data independent approach: $t=0$ case}\label{Appendix: data indep no filter}

\begin{lem}[Perturbation bound for Gaussian kernel] \label{lem:rounding}
For Gaussian kernel given by $K(\mathbf{x},\mathbf{y})=e^{-\frac{||\mathbf{x}-\mathbf{y}||_2^2}{2}}$, $\mathbf{x},\mathbf{y}\in \mathbb{R}^d$, if $\mathbf{x}\in \mathbb{R}^d$ and $\mathbf{x'}\in \mathbb{R}^d$ are such that $||\mathbf{x}-\mathbf{x'}||_2\le \alpha$, then:
\begin{align*}
    \max_{\mathbf{y}\in \mathbb{R}^d} |K(\mathbf{x},\mathbf{y})-K(\mathbf{x'},\mathbf{y})| \le \min(1,\frac{\alpha}{\sqrt{e}})
\end{align*}
\end{lem}
\begin{proof}
First, by triangle inequality and the assumption that $||x-x'||_2\le \alpha$, we have
\begin{align}\label{eq:triangleineq1}
    ||\mathbf{x}-\mathbf{y}||_2 - ||\mathbf{x'}-\mathbf{y}||_2 \le ||\mathbf{x}-\mathbf{x'}||_2 \le \alpha. 
\end{align}
For $f(x) = e^{-\frac{x^2}{2}}$, we have
\begin{align}\label{eq:maxderivative}
    \max_{x\in \mathbb{R}} |f'(x)| = \frac{1}{\sqrt{e}},
\end{align}
and thus
\begin{align*}
    \max_{\mathbf{y}\in \mathbb{R}^d}|K(\mathbf{x},\mathbf{y})-K(\mathbf{x'},\mathbf{y})|\le \frac{\alpha}{\sqrt{e}}.
\end{align*}
It remains to note that $0\leq K(\mathbf{x},\mathbf{y})\leq 1$ and thus also $\max_{\mathbf{y}\in \mathbb{R}^d}|K(\mathbf{x},\mathbf{y})-K(\mathbf{x'},\mathbf{y})|\le 1$.
\end{proof}

\begin{lem}[Error analysis of rounding to centers]\label{claim:rounding-centers}
Let $P\subset \mathbb{R}^d$ be the input dataset and let $P':=\{(c_1,v_1),\ldots,(c_J,v_J)\}$, where vector of $J$ point counts $\mathbf{v} \in \mathbb{R}^J$, and centers of $J$ bins $c_{1},\dots, c_{J}$ are defined in lines~\ref{line:v} and \ref{line:c} of Algorithm~\ref{alg:rounding-new}, respectively. Then for the KDE metric between $P'$ and $P$ we have
\begin{align}\label{eq:tripart1}
    \sup_{x\in \mathbb{R}^d}|\KD_{P'}(x) -\KD_{P}(x)|\le \max\left(\frac{w\sqrt{d}}{2\sqrt{e}},1\right).
\end{align}
\end{lem}
\begin{proof}%{Lemma~\ref{claim:rounding-centers}}
For any point $x\in \mathbb{R}^d$ that belongs to a bin with a center $c\in \mathbb{R}^d$ we have
\begin{align*}
    ||x-c||_2\le \frac{w\sqrt{d}}{2}.
\end{align*}
Plugging the above bound in Lemma~\ref{lem:rounding} completes the proof.
\end{proof}

\begin{lem}\label{cor:tailboundslaplace}
If total number of bins $J$ in Algorithm \ref{alg:rounding-new} is such that $J\leq n^{C}$, then for each bin $i\in[J]$, the noise term $\mathbf{w}_i=\mathbf{\wt{v}}_i-\mathbf{v}_i$ added in step \ref{line:laplace-noise} is such that
\begin{align*}
         \mathbf{w}_i \le \frac{4C\log n}{\epsilon},
     \end{align*}
     with probability at least $1-\frac{1}{n^C}$.
\end{lem}
\begin{proof}
    This is a consequence of Remark \ref{rem laplace tail log} and union bound argument.
\end{proof}
\begin{lem}[Bound on the total noisy count]\label{cor:v-tilde}
For $J$-dimensional vector of noisy point counts $\mathbf{\wt{v}}$ defined in line \ref{line:laplace-noise} of Algorithm \ref{alg:rounding-new}, with probability at least $1-\frac{1}{n^C}$ we have
\begin{align*}
    n - \frac{4CJ\log n}{\epsilon} \le |\mathbf{\wt{v}}|\le n+\frac{4CJ\log n}{\epsilon}.
\end{align*}
\end{lem}

\begin{proof}
    This is a consequence of Lemma \ref{cor:tailboundslaplace} and the fact that initial size of dataset is $n$.
\end{proof}

\begin{lem}[Error analysis of noise addition]\label{lem:data_ind_no_filter_Laplace_contribution}
Let $P':=\{(c_1,v_1),\ldots,(c_J,v_J)\}$, where vector of $J$ point counts $\mathbf{v} \in \mathbb{R}^J$, and centers of $J$ bins $c_{1},\dots, c_{J}$ are defined in lines~\ref{line:v} and \ref{line:c} of Algorithm~\ref{alg:rounding-new}, respectively. Let $Q:=\left\{(c_1,\wt{\mathbf{v}}_1),(c_2,\wt{\mathbf{v}}_2),\ldots,(c_J,\wt{\mathbf{v}}_J)\right\}$ be the noisy output of the algorithm. Then with probability at least $1-\frac{1}{n^{C}}$ we have
\begin{align}
     \sup_{x\in\mathbb{R}^d}|\KD_Q(x)-\KD_{P'}(x)| \le \frac{8CJ\log n}{\epsilon n - 4CJ\log n}.%\label{eq:Q2Pprime}
\end{align}
\end{lem}

\begin{proof}
For any $x\in \mathbb{R}^d$ we have
\begin{align*}
    |\KD_Q(x)-\KD_{P'}(x)| &= \left|\frac{1}{|\wt{\mathbf{v}}|}\sum_{i=1}^{J}\wt{\mathbf{v}}_i K(c_i,x) - \frac{1}{|{\mathbf{v}}|}\sum_{i=1}^{J}{\mathbf{v}}_i K(c_i,x)\right| \\
    % &= \frac{1}{|{\mathbf{{v}}}|}\left|\sum_{i=1}^{J}\frac{|{\mathbf{{v}}}|}{|{\mathbf{\wt{v}}}|}\wt{\mathbf{v}}_i K(c_i,q) - \sum_{i=1}^{J} {\mathbf{v}}_i K(c_i,q)\right|\\
    &=\frac{1}{n}\left|\sum_{i=1}^{J}K(c_i,x)\left(\frac{n}{|{\mathbf{\wt{v}}}|}\wt{\mathbf{v}}_i  - {\mathbf{v}}_i\right)  \right| %&&\text{since $|\mathbf{v}|=n$}
    \\
     &\le \frac{1}{n}\sum_{i=1}^{J}\left|\frac{n}{|{\mathbf{\wt{v}}}|}\wt{\mathbf{v}}_i  - {\mathbf{v}}_i \right| %&&\text{since $K(c_i,x) \le 1$},\\
    \end{align*}
where the second equality is the consequence of the fact that the point count in the original dataset is $n$, i.e. $|\mathbf{v}|=n$, and the inequality follows from $K(c_i,x) \le 1$. Let $\mathbf{w}_i$ denote the noise added to the $i$th bin's point count in step \ref{line:laplace-noise} of Algorithm \ref{alg:rounding-new}, so that $\mathbf{w}_i=\mathbf{\wt{v}}_i-\mathbf{v}_i$. Then we have
    \begin{align}
     |\KD_Q(x)-\KD_{P'}(x)| 
    %  &\le \frac{1}{n}\sum_{i=1}^{J}\left|\frac{n}{|{\mathbf{\wt{v}}}|}\left(\mathbf{v}_i+\mathbf{w}_i\right)  - {\mathbf{v}}_i \right|\\
     &\le \frac{1}{n}\sum_{i=1}^{J}\left|\mathbf{v}_i\left(\frac{n}{|{\mathbf{\wt{v}}}|}-1\right)  + \frac{n}{|{\mathbf{\wt{v}}}|}{\mathbf{w}}_i \right|\nonumber\\
     &\le \frac{1}{n}\sum_{i=1}^{J}\left|\mathbf{v}_i\left(\frac{n}{|{\mathbf{\wt{v}}}|}-1\right) \right| +\frac{1}{n}\sum_{i=1}^{J}\left| \frac{n}{|{\mathbf{\wt{v}}}|}{\mathbf{w}}_i \right|\\%\text{by triangle inequality}\nonumber\\
     &\le \left|\frac{n}{|{\mathbf{\wt{v}}}|}-1 \right| +\sum_{i=1}^{J}\left| \frac{1}{|{\mathbf{\wt{v}}}|}{\mathbf{w}}_i \right| &&%\text{since $|\mathbf{v}|=n$},\nonumber
     \end{align}
     where the second inequality is the triangle inequality and the last one follows as $|\mathbf{v}|=n$.
     Since $\mathbf{w}_i\sim\textsc{Lap}(2/\epsilon)$, by Lemma~\ref{cor:v-tilde} we have
     \begin{align}
     |\KD_Q(x)-\KD_{P'}(x)| &\le \left|\frac{n}{n - \frac{4CJ\log n}{\epsilon}}-1 \right| +\frac{1}{n - \frac{4CJ\log n}{\epsilon}}\sum_{i=1}^{J}\left| {\mathbf{w}}_i \right| \nonumber\\
     &\le \frac{n}{n - \frac{4CJ\log n}{\epsilon}}-1  +\frac{1}{n - \frac{4CJ\log n}{\epsilon}}\frac{4CJ\log n}{\epsilon}\\%&&\text{by Corollary~\ref{cor:tailboundslaplace}}\nonumber\\
    %  &\le \frac{4CJ\log n}{\epsilon n - 4CJ\log n}  +\frac{4CJ\log n}{\epsilon n - 4CJ\log n}\\
    &= \frac{8CJ\log n}{\epsilon n - 4CJ\log n}.\label{eq:Q2Pprime}
\end{align}
with probability $1-\frac{1}{n^{C}}$. The second inequality is the consequence of Lemma \ref{cor:tailboundslaplace}.
\end{proof}

\begin{proofof}{Theorem \ref{thm:general-1}}
    This is a consequence of Lemma \ref{claim:rounding-centers} and Lemma \ref{lem:data_ind_no_filter_Laplace_contribution} for $C$ such that $\delta = \frac{1}{n^{C}}$. Triangle inequality completes the proof.
\end{proofof}

\section{Data independant approach: $t>0$ case}\label{Appendix: data indep with filtering}

\begin{lem}[Algorithm~\ref{alg:rounding-new} filters out all $t/2$-light bins]\label{claim:lightbinssurviving}
If for $t=\frac{8C\log n}{\epsilon}$ and the total number of bins $J$ we have $J\leq n^{C}$ for some constant $C$, then with probability at least $1-\frac{1}{2}n^{-C}$ all $t/2$-light bins will be filtered out by step \ref{line:filtering} of Algorithm \ref{alg:rounding-new}.
\end{lem}

\begin{proof}
For $t=\frac{8C\log n}{\epsilon}$, since we are adding $\text{Lap}(2/\epsilon)$ noise the probability of a bin with point count less than $t/2$ having noisy point count more than $t$ is upper bounded by $\frac{1}{2} n^{-2C}$ (see Remark \ref{rem laplace tail log}).

Union bound over $J\leq n^C$ bins completes the proof.
\end{proof}
\begin{lem}[Algorithm~\ref{alg:rounding-new} does not filter any $3t/2$-heavy bins]\label{claim:heavybinsnotsurvive} 
If for $t=\frac{8C\log n}{\epsilon}$ and the total number of bins $J$ we have $J\leq n^{C}$ for some constant $C$, then with probability at least $1-\frac{1}{2}n^{-C}$ no $3t/2$-heavy bin gets filtered out by step \ref{line:filtering} of Algorithm \ref{alg:rounding-new}.
 Algorithm~\ref{alg:rounding-new} does not filter any $3t/2$-heavy bin with probability at least
    $1-\frac{1}{2}n^{-C}$.
\end{lem}

\begin{proof}
For $t=\frac{8C\log n}{\epsilon}$, since we are adding $\text{Lap}(2/\epsilon)$ noise the probability of a bin with point count at least $3t/2$ having noisy point count less than $t$ is upper bounded by $\frac{1}{2} n^{-2C}$ (see Remark \ref{rem laplace tail log}). Union bound argument over $J<n^{C}$ bins completes the proof.
\end{proof}

\begin{lem}[Noisy point counts]\label{claim:v-tilde-alg-filtering}
If for $t=\frac{8C\log n}{\epsilon}$ and the total number of bins $J$ we have $J\leq n^{C}$ for some constant $C$, then for $J$-dimensional vector of noisy point counts $\mathbf{\wt{v}}$ defined in line \ref{line:laplace-noise} of 
Algorithm~\ref{alg:rounding-new}, with probability at least $1-\frac{1}{n^C}$ we have
\begin{align*}
    n -m- \frac{4CM\log n}{\epsilon} \le |\mathbf{\wt{v}}|\le n+\frac{4CM\log n}{\epsilon},
\end{align*}
where $M$ and $m$ denote the total number of $t/2$-heavy bins and the total number of points in $3t/2$-light bins, respectively.
\end{lem}

\begin{proof}
Let $F:=\{i: \mathbf{v}_i>0, \mathbf{\wt{v}}_i = 0\}$, $Z:=\{i: \mathbf{v}_i=0, \mathbf{\wt{v}}_i = 0\}$ and $H:=\{i: \mathbf{\wt{v}}_i>0\}$ denote the set of non empty bins that are filtered out, the set of empty bins that are filtered out and the set of bins that survive filtering, respectively. Note that every bin belongs to one of the three sets i.e. $[J] = F \cup Z\cup H$. We have
\begin{align*}
    |\mathbf{\wt{v}}| &= \sum_{i = 1}^{J} \mathbf{\wt{v}}_i\\
    &= \sum_{i\in F}\mathbf{\wt{v}}_i +\sum_{i\in Z}\mathbf{\wt{v}}_i +\sum_{i\in H}\mathbf{\wt{v}}_i \\
    &= \sum_{i\in H}\mathbf{\wt{v}}_i  \\
    &\le |\mathbf{v}| + \left|H\right|\cdot \frac{4C\log n}{\epsilon}\\
    &\le n + \frac{4CM\log n}{\epsilon},
\end{align*}
where the third equality follows by definition of $F$ and $Z$, and the first inequality is the consequence of Lemma \ref{cor:tailboundslaplace}. The last inequality follows from $|\mathbf{v}|=n$ and the consequence of Lemma \ref{claim:lightbinssurviving} which gives that with probability at least $1-\frac{1}{2}n^{-C}$ any bin that survives filtering is $t/2$-heavy i.e. $|H|\le M$.
On the other hand, we also have
\begin{align*}
    |\mathbf{\wt{v}}| &= \sum_{i = 1}^{J} \mathbf{\wt{v}}_i\\
    &= \sum_{i\in H}\mathbf{\wt{v}}_i \\
    &\ge \sum_{i\in H}\mathbf{v}_i - |H|\cdot\frac{4C\log n}{\epsilon}\\
    &\ge \sum_{i\in H}\mathbf{v}_i - \frac{4CM\log n}{\epsilon}\\
    &\ge n-m - \frac{4CM\log n}{\epsilon}
\end{align*}
where again second equality comes from the definition of $F$ and $Z$, and the first inequality is the consequence of Lemma \ref{cor:tailboundslaplace} and the second inequality follows from $|H|\le M$ as above. Finally, the last inequality is the consequence of Lemma \ref{claim:heavybinsnotsurvive} which gives us that with probability at least $1-\frac{1}{2}n^{-C}$ any bin that gets filtered out is $3t/2$-light and so the total number of filtered out points is upper bounded by $m$. This means that the total number of points in bins that survive filtering is at least $n-m$ i.e. $\sum_{i\in H} \mathbf{v}_i \ge |\mathbf{v}|- m \ge n - m$. Union bound argument completes the proof.
\end{proof}

Now, we analyze the error between kernel density induced by $Q$ and $P'$. For any $q\in \mathbb{R}^d$ we have

\begin{lem}\label{lem appendix Q P' t>0}
     Let $P':=\{(c_1,v_1),\ldots,(c_J,v_J)\}$, where $v$, $c$ are defined as in lines ~\ref{line:v} and \ref{line:c} of Algorithm~\ref{alg:rounding-new}. For $t=\frac{8C\log n}{\epsilon}$ and the total number of bins $J\leq n^{C}$, with probability $1-n^{-C}$ we have

     \begin{align}
     \sup_{x\in\mathbb{R}^{d}}|\KD_Q(x)-\KD_{P'}(x)| \leq \frac{\epsilon m + 8CM\log n}{\epsilon n - \epsilon m - 4CJ\log n}+\frac{m}{n}.
\end{align}
\end{lem}

\begin{proof}
Let $H:=\{i: \mathbf{\wt{v}}_i>0\}$ denote the bins that survive filtering step. We have
\begin{align*}
    |\KD_Q(x)-\KD_{P'}(x)| &= \left|\frac{1}{|\wt{\mathbf{v}}|}\sum_{i=1}^{J}\wt{\mathbf{v}}_i K(c_i,x) - \frac{1}{|{\mathbf{v}}|}\sum_{i=1}^{J}{\mathbf{v}}_i K(c_i,x)\right| \\
    % &= \frac{1}{|{\mathbf{{v}}}|}\left|\sum_{i=1}^{J}\frac{|{\mathbf{{v}}}|}{|{\mathbf{\wt{v}}}|}\wt{\mathbf{v}}_i K(c_i,q) - \sum_{i=1}^{J} {\mathbf{v}}_i K(c_i,q)\right|\\
    &=\frac{1}{n}\left|\sum_{i=1}^{J}K(c_i,x)\left(\frac{n}{|{\mathbf{\wt{v}}}|}\wt{\mathbf{v}}_i  - {\mathbf{v}}_i\right)  \right|\\
     &\le \frac{1}{n}\sum_{i=1}^{J}\left|\frac{n}{|{\mathbf{\wt{v}}}|}\wt{\mathbf{v}}_i  - {\mathbf{v}}_i \right|\\
     &\le \frac{1}{n}\left(\sum_{i\in H}\left|\frac{n}{|{\mathbf{\wt{v}}}|}\wt{\mathbf{v}}_i  - {\mathbf{v}}_i \right| + \sum_{i\not\in H} |\mathbf{v}_i|\right)\\
     &\le\frac{1}{n}\left(\sum_{i\in H}\left|\frac{n}{|{\mathbf{\wt{v}}}|}\wt{\mathbf{v}}_i  - {\mathbf{v}}_i \right| + m\right)
    \end{align*}
where the second equality holds since $|\mathbf{v}|=n$, and the first inequality is the consequence of $K(c_i,x)\leq 1$. The second inequality follows by the definition of $H$, and the final one is the consequence of Lemma \ref{claim:heavybinsnotsurvive}, as with probability at least $1-\frac{1}{2}n^{-C}$ any bin that is filtered out must be $3t/2$-light. Let $\mathbf{w}_i=\mathbf{\wt{v}}_i-\mathbf{v}_i$, then $\mathbf{w}_i\sim\textsc{Lap}(2/\epsilon)$. Then we further have
    \begin{align}
     |\KD_Q(x)-\KD_{P'}(x)| 
     &\le\frac{1}{n}\left(\sum_{i\in H}\left|\frac{n}{|{\mathbf{\wt{v}}}|}\wt{\mathbf{v}}_i  - {\mathbf{v}}_i \right| + m\right)\\
     &\le \frac{1}{n}\left(\sum_{i\in H}\left|\mathbf{v}_i\left(\frac{n}{|{\mathbf{\wt{v}}}|}-1\right)  + \frac{n}{|{\mathbf{\wt{v}}}|}{\mathbf{w}}_i \right|+m\right)\nonumber\\
     &\le \frac{1}{n}\left(\sum_{i\in  H}\left|\mathbf{v}_i\left(\frac{n}{|{\mathbf{\wt{v}}}|}-1\right) \right| +\frac{1}{n}\sum_{i\in H}\left| \frac{n}{|{\mathbf{\wt{v}}}|}{\mathbf{w}}_i \right|+m\right)\nonumber\\
     &\le \left|\frac{n}{|{\mathbf{\wt{v}}}|}-1 \right| +\sum_{i\in H}\left| \frac{1}{|{\mathbf{\wt{v}}}|}{\mathbf{w}}_i \right|+\frac{m}{n}\nonumber
     \end{align}
where the third is the triangle inequality and the last one follows since $|\mathbf{v}|=n$. Let $n':=n -m- \frac{4CM\log n}{\epsilon}$. By Lemma \ref{claim:v-tilde-alg-filtering} with probability $1-n^{-C}$ we have
     \begin{align}
     |\KD_Q(x)-\KD_{P'}(x)| &\le \left|\frac{n}{n'}-1 \right| +\frac{1}{n'}\sum_{i=1}^{J}\left| {\mathbf{w}}_i \right| + \frac{m}{n} \nonumber\\
     &\le \frac{n}{n'}-1  +\frac{4CM\log n}{n'\epsilon}+\frac{m}{n}\nonumber\\
    %  &\le \frac{4CJ\log n}{\epsilon n - 4CJ\log n}  +\frac{4CJ\log n}{\epsilon n - 4CJ\log n}\\
    &= \frac{\epsilon m + 8CM\log n}{\epsilon n - \epsilon m - 4CJ\log n}+\frac{m}{n}.\label{eq:Q2Pprime-alg-filtering}
\end{align}
where the second inequality holds by Lemma \ref{cor:tailboundslaplace} 
\end{proof}

\begin{proofof}{Theorem \ref{thm:general-2}}

As a consequence of triangle inequality and Lemmas \ref{claim:rounding-centers} and \ref{lem appendix Q P' t>0}, for $C$ such that $\delta = n^{-C}$ we have
\begin{align*}
\sup_{x\in\mathbb{R}^{d}}|\KD_P(q) - \KD_Q(q)|&\le \frac{\epsilon m + 8CM\log n}{\epsilon n - \epsilon m - 4CM\log n} +\frac{m}{n}+ \frac{w\sqrt{d}}{2\sqrt{e}}\\
    &= \frac{\epsilon m + 8M\log \frac{1}{\delta}}{\epsilon n - \epsilon m - 4M\log \frac{1}{\delta}} +\frac{m}{n}+\frac{w\sqrt{d}}{2\sqrt{e}}.
\end{align*}
This completes the proof.
\end{proofof}

\section{Special case: Original dataset from mixture of Gaussians}\label{Appendix Gausssian}

\begin{lem}\label{lem aux furtherst heavy gaussian}
    If $r>0$, $C>0$ are such that %$\frac{n w^d}{(2\pi\sigma^2)^{d/2}}e^{-\frac{(r/2)^2}{2\sigma^2}}> \frac{2C\log n}{\epsilon}$
    $\frac{n w^d}{(2\pi\sigma^2)^{d/2}}e^{-\frac{(r/2)^2}{2\sigma^2}}> \frac{2C\log n}{\epsilon}$,
    then $r\leq 3\sigma\sqrt{\log n + d \log  \left(\frac{w}{\sigma\sqrt{2\pi}}\right)}$.
\end{lem}

\begin{proof}
    The condition of the lemma translates to

\begin{align*}
    r &\le 2\sqrt{2}\sigma\sqrt{\log\left(\frac{ n}{2C\log n} \cdot \left(\frac{w}{\sigma\sqrt{2\pi}}\right)^d \right)}\\
    &= 2\sqrt{2}\sigma \sqrt{\log n - \log\log n -\log 2C + d\log \left(\frac{w}{\sigma\sqrt{2\pi}}\right)} \\
    &\le 3\sigma\sqrt{\log n + d \log  \left(\frac{w}{\sigma\sqrt{2\pi}}\right)} 
    % &\ge \sigma \sqrt{\frac{\log n}{2}}.
\end{align*}
\end{proof}

\begin{lem}\label{lem aux light gaussian}
  If $r>0$, $C>0$ are such that $\frac{n w^d}{(2\pi\sigma^2)^{d/2}}e^{-\frac{r^2}{\sigma^2}}< \frac{16C\log n}{\epsilon}$, then we have
  $r\geq \sigma\sqrt{\log (\epsilon n) - \log(16 C \log n)+ d\log\left(\frac{w}{\sigma \sqrt{2\pi}}\right)}$
\end{lem}

\begin{proof}
Condition on $r$ translates to
\begin{align*}
    r 
     &\ge \sqrt{-\sigma^2\log\left(\frac{16C\log n}{\epsilon n} \cdot \left(\frac{\sigma\sqrt{2\pi}}{w}\right)^d \right) }\\
    &\ge\sigma\sqrt{\log\left(\frac{\epsilon n}{16C\log n} \cdot \left(\frac{w}{\sigma\sqrt{2\pi}}\right)^d \right)}\\
    &= \sigma \sqrt{\log (\epsilon n) - \log(16C\log n) + d\log \left(\frac{w}{\sigma\sqrt{2\pi}}\right)}.
\end{align*}
\end{proof}

\begin{defn}
For a bin (hypercube) $B$ and a point $y\in \mathbb{R}^d$, we define their distance as the $\ell_2$ distance of the 
center of $B$ to $y$.
\end{defn}
\begin{lem}[Upper bound on $M$]     
     For a dataset coming from a multivariate Gaussian with variance $\sigma^2I$ in $\mathbb{R}^d$  there are at most $$\left(\frac{6\sigma}{w}\sqrt{\log n + d \log  \left(\frac{w}{\sigma\sqrt{2\pi}}\right)} + 2\right)^d$$
     $t/2$-heavy bins with arbitrary high probability.
\end{lem}
\begin{proof}

Without loss of generality we assume that the mean of the distribution is the origin. Let $B$ be a bin at distance $r$ from the origin, and $x\in\mathbb{R}^d$ be a point inside $B$. For $X$ from multivariate Gaussian we have

\begin{align*}
         %X\sim N(0,\sigma^2I) \rightarrow 
f(X=x)\le\frac{1}{(2\pi\sigma^2)^{d/2}}e^{-\frac{(r - \frac{w\sqrt{d}}{2})^2}{2\sigma^2}}.
\end{align*}

Hence, the expected number of points within $B$ is upper bounded by $\frac{n w^d}{(2\pi\sigma^2)^{d/2}}e^{-\frac{(r - \frac{w\sqrt{d}}{2})^2}{2\sigma^2}}$. For $r$ such that $w\sqrt{d}\le r$, this is further upper bounded by $\frac{n w^d}{(2\pi\sigma^2)^{d/2}}e^{-\frac{(r/2)^2}{2\sigma^2}}$.

For simplicity of notation, let us introduce $\mu(r) = \frac{n w^d}{(2\pi\sigma^2)^{d/2}}e^{-\frac{(r/2)^2}{2\sigma^2}}$. If we assume that $r$ is large enough so that %$\mu(r) \le \frac{2C\log n}{\epsilon}$, 
$\mu(r) \le {2C\log n}$,
Chernoff bounds (see Lemma~\ref{lem:chernoff}) give us
\begin{align*}
    \Pr\left[| B|\ge 
    4C\log n\right]&\le 2e^{-\mu}\\
    &\le \frac{1}{n^{2C}},
\end{align*}
where $|B|$ denotes the number of the points from the dataset within $B$. Thus if $r$ is such that $\mu(r)> 2C\log n$, then bins at distance at least $r$ from the origin are $t/2$-light with probability at least $1-\frac{1}{n^{2C}}$. Equivalently this means that all $t/2$-heavy bins are at distance at most $r$ from the origin. 
 
By Lemma \ref{lem aux furtherst heavy gaussian} the furthest bin that can be $t/2$-heavy is at distance at most $$3\sigma\sqrt{\log n + d \log  \left(\frac{w}{\sigma\sqrt{2\pi}}\right)}$$
from the origin.
It remains to note that the area of diameter $\alpha$ is covered by at most $\alpha/w+2$ bins of width on a single axis. Thus, the number of $t/2$-heavy bins is bounded by
$$\left(\frac{6\sigma}{w}\sqrt{\log n + d \log  \left(\frac{w}{\sigma\sqrt{2\pi}}\right)} + 2\right)^d$$ 
with probability at least $1-\frac{1}{n^{2C}}$. 
\end{proof}

\begin{lem}[Upper bound on $m$]
For a dataset coming from a multivariate Gaussian with variance $\sigma^2I$ in $\mathbb{R}^d$, when binning is done with widths such that $8\le \log (\frac{w}{\sigma\sqrt{2\pi}})$, there are at most $n^{2/3+o(1)}$ points in
$3t/2$-light bins with arbitrary high probability.
\end{lem}

\begin{proof}
Without loss of generality we assume that the mean of the distribution is the origin. Let $B$ be a bin at distance $r$ from the origin, and $x\in\mathbb{R}^d$ be a point inside $B$. For $X$ from multivariate Gaussian we have

\begin{align*}
         %X\sim N(0,\sigma^2I) \rightarrow 
         \frac{1}{(2\pi\sigma^2)^{d/2}}e^{-\frac{(r + \frac{w\sqrt{d}}{2})^2}{2\sigma^2}}\le f(X=x).
\end{align*}
Hence, the expected number of points within $B$ is lower bounded by $\frac{n w^d}{(2\pi\sigma^2)^{d/2}}e^{-\frac{(r + \frac{w\sqrt{d}}{2})^2}{2\sigma^2}}$. For $r$ such that $\frac{w\sqrt{d}}{2}\le (\sqrt{2}-1)r$ we further have the lower bound of $\mathbb{E} [\vert B \vert]\geq \frac{n w^d}{(2\pi\sigma^2)^{d/2}}e^{-\frac{2r^2}{2\sigma^2}}$, where $B$ denotes the number of points from the dataset within $B$.

For simplicity of notation let $\mu(r) = \frac{n w^d}{(2\pi\sigma^2)^{d/2}}e^{-\frac{r^2}{\sigma^2}}$. For $r$ such that $\mu(r) \ge \frac{16C\log n}{\epsilon}$, 
%$\mu(r) \ge {16C\log n}$, 
Chernoff bounds (see Lemma~\ref{lem:chernoff}) we have

\begin{align*}   \Pr\left[|B|\le \frac{{12C\log n}}{\epsilon} \right]
    &\le \frac{1}{n^{2C}}.
\end{align*}
Thus if $r$ is such that $\mu(r)\ge \frac{16C\log n}{\epsilon}$ 
%$\mu(r)\ge {16C\log n}$
for some $r$, then bins at distance at most $r$ are $3t/2$-heavy with probability at least $1-\frac{1}{n^{2C}}$. In other words, all $3t/2$-light bins are at distance at least $r$.
By Lemma \ref{lem aux light gaussian}, the smallest $r$ such that all $3t/2$-light bins are at distance at least $r$ is % OLD $\frac{1}{2}\sigma\sqrt{\log n + d \log  \left(\frac{w}{\sigma\sqrt{2\pi}}\right)}$.
$\sigma\sqrt{\log (\epsilon n) - \log(16 C \log n)+ d\log\left(\frac{w}{\sigma \sqrt{2\pi}}\right)}$.

From Corollary~\ref{cor:tailboundschisquared} it follows that the probability of a point taking distance at least %$\frac{1}{2}\sigma\sqrt{\log n + d \log  \left(\frac{w}{\sigma\sqrt{2\pi}}\right)}$ 
$\sigma\sqrt{\log (\epsilon n) - \log(16 C \log n)+ d\log\left(\frac{w}{\sigma \sqrt{2\pi}}\right)}$
from the cluster center is

\begin{align*}
    \Pr\left[X \ge \sigma \sqrt{2d + 3 \left(\frac{\log (\epsilon n) -\log(16C \log n)+ d \log  (\frac{w}{\sigma\sqrt{2\pi}})-2d}{3}\right)}\right] &\le \exp\left({-\frac{(\log (\epsilon n) -\log(16C \log n)+ d (\log  (\frac{w}{\sigma\sqrt{2\pi}})-2)}{3}}\right)\\
    & \leq \frac{(16 C\log n)^{1/3}\cdot e^{-\frac{d}{3}(\log \frac{w}{\sigma \sqrt{2 \pi}}-2)}}{(\epsilon n)^{1/3}}
\end{align*}

Let $\mathcal{C}'$ be the set of points %in $\mathcal{C}$ 
that are at least %$\sigma \sqrt{2d + 3 \left(\frac{\frac{1}{4}\log n + \frac{1}{4}d \log  (w/\sigma)-2d}{3}\right)}$ 
$\sigma \sqrt{2d + 3 \left(\frac{\log (\epsilon n) -\log(16C \log n)+ d \log  (\frac{w}{\sigma\sqrt{2\pi}})-2d}{3}\right)}$
far from the mean of the ditribution i.e. origin. Any point in $3t/2$-light bins belongs to $\mathcal{C}'$ with probability at least $1-\frac{1}{n^{2C}}$. Chernoff bound (see Lemma~\ref{lem:chernoff}) gives us
\begin{align*}
    |\mathcal{C}'|\le \epsilon^{-1/3} n^{2/3}{(16 C\log n)^{1/3}\cdot e^{-\frac{d}{3}(\log \frac{w}{\sigma \sqrt{2 \pi}}-2)}}
\end{align*}
with high probability. This completes the proof.
\end{proof}

\begin{proofof}{Theorem \ref{thm one gaussian tradeoff}}
From above analysis we have that the number of $t/2$-heavy bins is bounded by
$$M = \left(\frac{6\sigma}{w}\sqrt{\log n + d \log  \left(\frac{w}{\sigma\sqrt{2\pi}}\right)} + 2\right)^d$$ 
with high probability.
We also have that the total number of points in $3t/2$-light bins is upper bounded by %$m=n^{11/12}+o(1)$. 
$m = \epsilon^{-1/3} n^{2/3}{(16 C\log n)^{1/3}\cdot e^{-\frac{d}{3}(\log \frac{w}{\sigma \sqrt{2 \pi}}-2)}}$. Thus  we have
\begin{align}
M &= \left(\frac{6\sigma}{w}\sqrt{\log n + d \log  \left(\frac{w}{\sigma\sqrt{2\pi}}\right)} + 2\right)^d \\
& \leq \left(\frac{6\sigma}{w}\sqrt{2}\sqrt{\log n} +2\right)^d\\
& \leq \left(\frac{12\sigma}{w}\right)^d (\log n)^{d/2}
\end{align}
where the first inequality follows by $n\geq \left(\frac{w}{\sigma \sqrt{2\pi}}\right)^d$ and the second also follows for large $n$. We have
\begin{align}
    \epsilon m + 8M\log \frac{1}{\delta} &\leq \epsilon m +8 \log \left(\frac{1}{\delta}\right) \cdot \left(\frac{12\sigma}{w}\right)^d (\log n)^{d/2}\\
    & \leq \frac{1}{2}\epsilon n
\end{align}
where the second inequality follows from condition $\frac{n}{(\log n)^{d/2}}>16\log \left(\frac{1}{\delta}\right)\cdot \left(\frac{12\sigma}{w}\right)^d$.
Thus remains to apply Theorem \ref{thm:general-2} and we get

\begin{align}
\frac{\epsilon m + 8M\log \frac{1}{\delta}}{\epsilon n - \epsilon m - 4M\log \frac{1}{\delta}} +\frac{m}{n} + \frac{w\sqrt{d}}{2\sqrt{e}} &\leq \frac{\epsilon(\epsilon^{-1/3} n^{2/3}{(16 C\log n)^{1/3}\cdot e^{-\frac{d}{3}(\log \frac{w}{\sigma \sqrt{2 \pi}}-2)}})+8 \log \left(\frac{1}{\delta}\right) \cdot \left(\frac{12\sigma}{w}\right)^d (\log n)^{d/2}}{\frac{1}{2}\epsilon n}\\
&+\frac{n^{2/3}{(16 C\log n)^{1/3}\cdot e^{-\frac{d}{3}(\log \frac{w}{\sigma \sqrt{2 \pi}}-2)}}}{n}+ \frac{w\sqrt{d}}{2\sqrt{e}}\\
& \leq \frac{3 (16 C\log n)^{1/3}\cdot e^{-\frac{d}{3}(\log \frac{w}{\sigma \sqrt{2 \pi}}-2)}}{(\epsilon n)^{1/3}}+\frac{16 \log \left(\frac{1}{\delta}\right)\cdot \left(\frac{12 \sigma}{w}\right)^d (\log n)^{d/2}}{\epsilon n}+\frac{w\sqrt{d}}{2\sqrt{e}}
\end{align}
\end{proofof}

\section{Data dependent algorithm}

\begin{figure}[h]
    \centering
\includegraphics[width=0.3\textwidth]{figures/step1.pdf}\hfill
\includegraphics[width=0.3\textwidth]{figures/step2.pdf}\hfill
\includegraphics[width=0.3\textwidth]{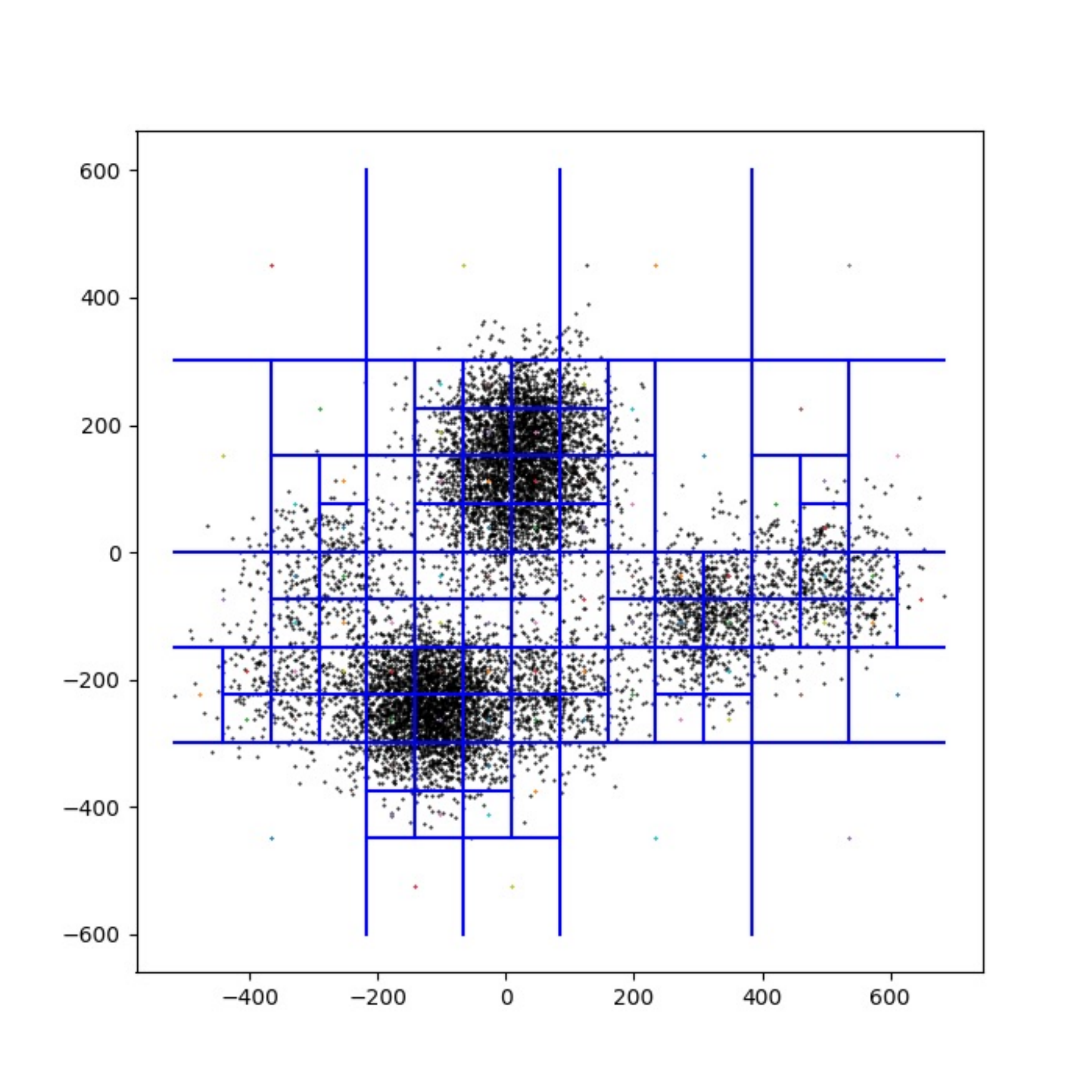}
    \caption{Example of stages of data dependent partitioning of the dataset in $\mathbb{R}^2$.}
    \label{fig:binning}
\end{figure}

\subsection{Implicit sampling for data independent Algorithm \ref{alg:rounding-new}}\label{appendix secti implicit data indep}

For the data independent algorithm, implicit implementation of empty bins sampling is straightforward. Instead of storing $(\frac{R}{w})^{d}$ bins centers, it is enough to store only those corresponding to non empty bins, and sample empty centers via independent uniform random sampling of each coordinate from the set of possible values (and rejection if gluing them together gives center corresponding to a non empty bin). Note that this requires storing $O(\frac{R}{w}\cdot d)$ values for non empty bins instead of $O(\frac{R}{w})^{d}$.

\section{Implicit sampling of empty bins} \label{appendix sect implicit}

\begin{figure}
\centering
\scalebox{0.5}{
\begin{forest}
for tree={
    grow=south,
    circle, draw, minimum size=3ex, inner sep=0.75pt,
    s sep=2mm
        }
[
    [
        [, edge=dashed,fill=gray
        [,no edge, draw=none]
        [,no edge, draw=none]
        ]
        [ 
            [ 
            [
            [,fill=black[,no edge, draw=none]
        [,no edge, draw=none]]
            [, edge=dashed,fill=gray
            [,no edge, draw=none
            ]
        [,no edge, draw=none]]
            ]
            [, edge=dashed,fill=gray
            [,no edge, draw=none
            ]
        [,no edge, draw=none]]
            ]
            [, edge=dashed,fill=gray
            [,no edge, draw=none]
        [,no edge, draw=none]
        ]
        ]
    ]
    [
        [
            [, edge=dashed,fill=gray
            [,no edge, draw=none]
        [,no edge, draw=none]
        ]
            [,fill=black
            [,no edge, draw=none]
        [,no edge, draw=none]] 
        ]
        [, edge=dashed,fill=gray
        [,no edge, draw=none]
        [,no edge, draw=none]
        ]
    ]
]
\end{forest}}
\vspace{-0.5cm}
\caption{Tree with $h = 2$ and $h' = 5$. Black nodes are non-empty bins, gray nodes are empty bins we need to sample.}
\end{figure}

\begin{proofof}{Lemma \ref{lem:sampling explicit equivalent implicit}}
Equivalence of Algorithm \ref{alg:explicit-empty-bins} and Algorithm \ref{alg:implicit-empty-bins} is the consequence of independence of Bernoulli indicators (as Laplace noise are independent for different bins) and the fact that Binomial can be represented as the sum of independent Bernoullis.
\end{proofof}

\begin{lem} \label{lem total number empty bins}
For a dataset of size $n$, if the decision tree is data independent up to depth level $h$ and data dependent in the remaining part, then the number of empty bins is upper bounded by
\begin{align*}
    2^h + n(h'-h)
\end{align*}
where $h'$ denotes the total number of levels. 
\end{lem}
\begin{proof}
We need to upper bound the number of leaves for such binary tree. Since the binary tree is complete up to depth $h$, we have at most $2^{h}$ leaves at level $h$. Furthermore, for any non empty bin we can have at most $h'-h$ empty bins between depth $h+1$ and $h'$. Thus, we have at most $n(h'-h)$ empty bins between depth $h+1$ and $h'$.
\end{proof}

\begin{proofof}{Lemma \ref{lem:empty bin not split}}
For a single empty bin, by Remark \ref{fact:tailbound4laplace} we have 
\begin{align*}
         \Pr[\text{Lap}(\frac{2(h'-h)}{\epsilon}) \ge \tau] = \frac{1}{2}e^{-\frac{\tau}{2(h'-h)/\epsilon}}.
     \end{align*}
Thus for $$\tau = \frac{2(h'-h)}{\epsilon}\log \left(\frac{1}{\delta}\cdot\left(2^h + n(h'-h)\right)\right)$$ the right hand side is less than $\frac{\delta}{2^h + n(h'-h)}$. As a consequence of Lemma \ref{lem total number empty bins} there are at most $2^h + n(h'-h$ empty bins and thus union bound argument completes the proof.
\end{proofof}

\paragraph{More implementation details:} Our algorithm for implicitly sampling the empty bins proceeds as follows: in the recursion tree, from left to right, our algorithm first finds common ancestor for any two consecutive non-empty bins, say $i$'th and $i+1$'th and calls it \emph{$i$'th common ancestor}. Then, it calculates the number of empty bins between $i$'th non-empty bin and $i$'th common ancestor. And similarly, it calculates the number of empty bins between $i+1$'th non-empty bin and $i$'th common ancestor. Note that by the above-mentioned binomial distribution implicit sampling argument, one can apply a binomial distribution sampling technique to each of these numbers, and it is not hard to locate the sampled empty bins. Finally, we need to apply an implicit sampling technique to empty bins at level $h$. This is done similarly to the rejection sampling technique we mentioned in implicit implementation of data independent algorithm.

\section{ Experiments}\label{appendix experiments}
\subsection{Experimental setting for comparison with \cite{DBLP:conf/icml/BalogTS18}} \label{appendix_balog_exp_setting}

For both dimension $2$ and $5$, the dataset consists of $n = 100, 000$ samples from
a multivariate mixture of Gaussians. The mixture has $10$ components, with mixing weights proportional to $(1, 1/2,\dots, 1/10)$, and their means are chosen from spherical Gaussian with mean $[100,\dots,100]$ and covariance $200I$.  Each point is simulated by first
sampling the mixture component, and then sampling from a spherical Gaussian centered at the mean of the chosen mixture component and with covariance $30I$. For accuracy of the comparison, we did not re-sample the dataset and used the exact version in the code of \cite{DBLP:conf/icml/BalogTS18}.

\subsection{Dependency on dataset size}\label{appendix: dep on size}
Intuitively, it is easier to hide the contribution of an individual in a large set, compared to a small set, in kernel density. We show this empirically using three datasets with different dataset sizes, but the same underlying mixture of Gaussians parameters as the 5 dimensional datasets in the experiments section, see Figure~\ref{Fig:dep-N}. This experiment also shows that our algorithm is able to produce synthetic datasets with better minimum error when the size of the dataset, $N$, is larger. This dependency in $N$ was also evident in our theoretical results in Theorem~\ref{thm:general-2} and Theorem~\ref{thm one gaussian tradeoff}.
\begin{figure}[ht!]
    \centering
    \includegraphics[scale=0.33]{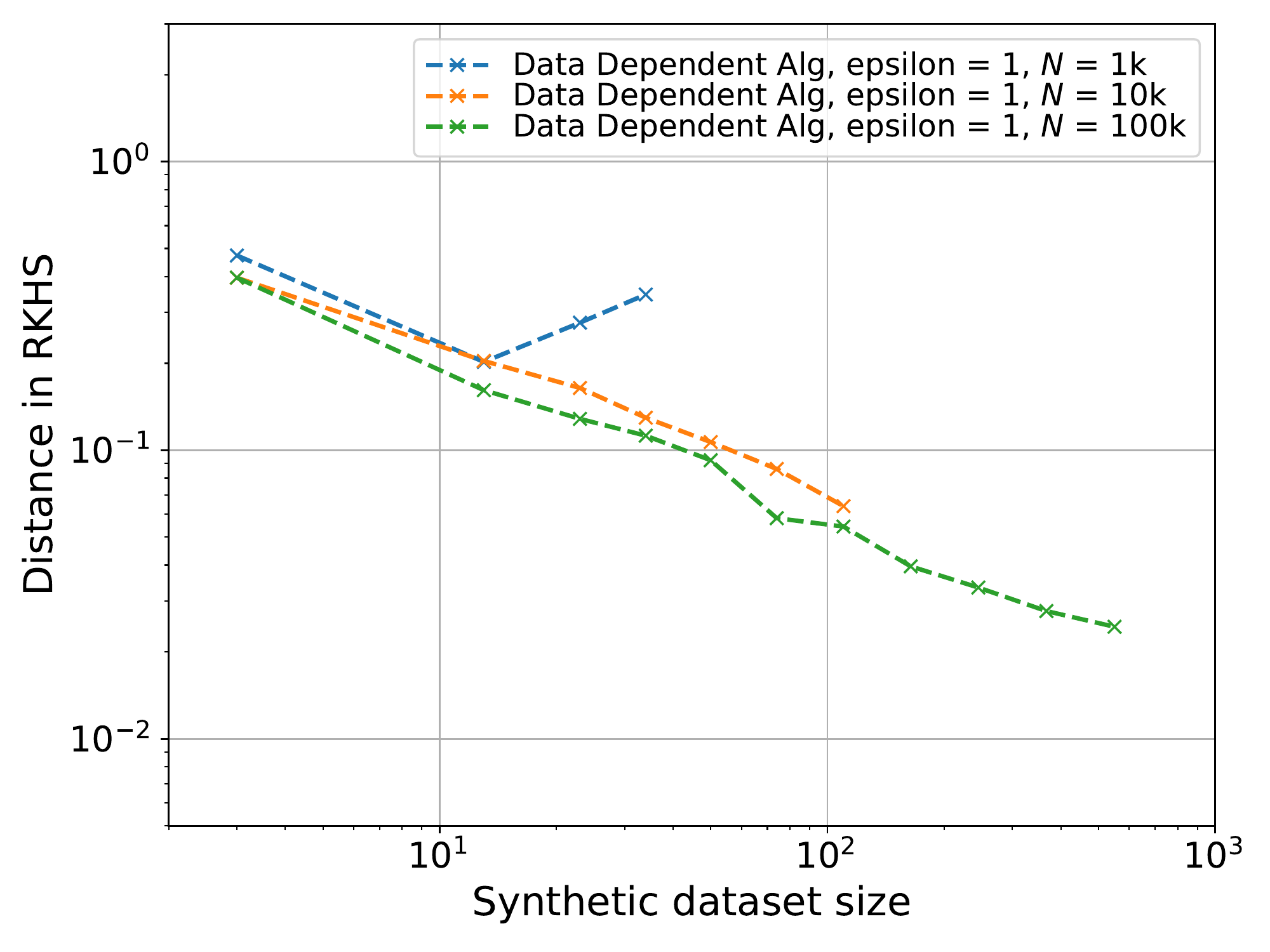}
    \caption{Using three 5 dimensional datasets with 1k, 10k and 100k data points, yet with the same underlying distribution parameters, we show that larger dataset size, $N$, naturally results in a better performance. Moreover, for larger $N$, our algorithm is capable of achieving better minimum error.}
    \label{Fig:dep-N}
\end{figure}

\subsection{Dependency on variance}\label{appendix: dep on variance}

Next, we show the effect of $\sigma$ (see Theorem~\ref{thm one gaussian tradeoff} for the definition of $\sigma$) in the mixture of Gaussians datasets. We consider two datasets with underlying mixture of Gaussians distributions in dimension $10$, with $\sigma = 30$ and $\sigma = 3$ for each cluster. As expected, the simulation results presented in Figure~\ref{fig:d10-sigma} confirm that our algorithm performs better in the setting where clusters are more concentrated around a center, i.e., small $\sigma$ case.

\begin{figure}[ht!]
    \centering
    \includegraphics[width=0.33\textwidth]{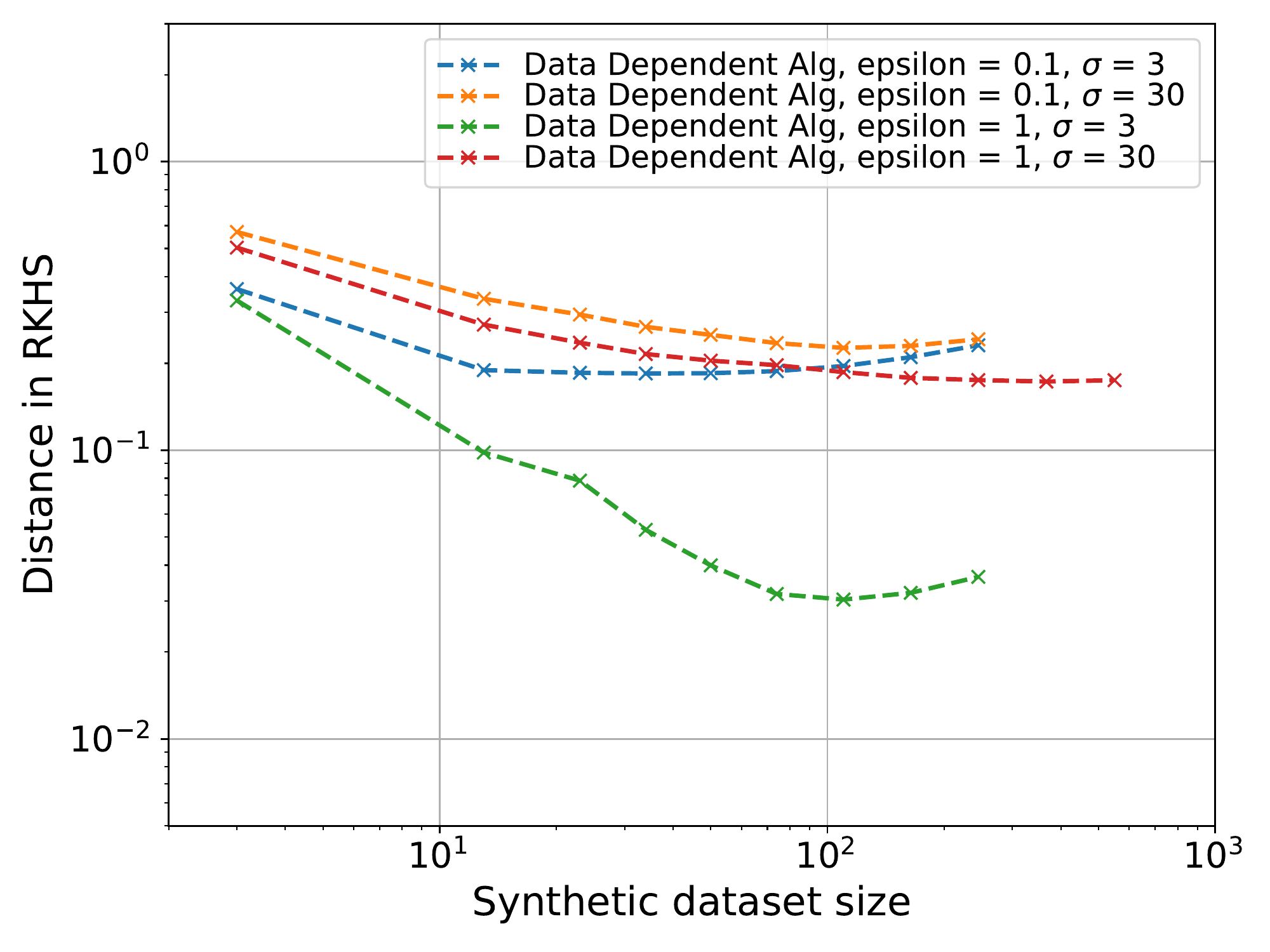}
    \caption{Performance comparison of our data dependent algorithm on 10-dimensional datasets with underlying mixture of Gaussians distribution with $\sigma = 3$ and $\sigma = 30$. Note that our algorithm performs better with smaller $\sigma$ as predicted by our theory.}
    \label{fig:d10-sigma}
\end{figure}

\subsection{Binary classification}\label{appendix:binary_classification}

We use a dataset from a Kaggle competition \url{https://www.kaggle.com/datasets/mlg-ulb/creditcardfraud} with information of credit card transactions which were either
fraudelent or not, which was also used in \cite{DBLP:conf/aistats/HarderAP21}. This dataset has 31 categories, 30 numerical features and a binary label. Similarly to \cite{DBLP:conf/aistats/HarderAP21}, we use all but the first feature (Time).

For both our data dependent algorithm and DP-MERF \citep{DBLP:conf/aistats/HarderAP21}, we use $80\%$ of input data for synthetic data generation, for various privacy budgets $\epsilon$. Synthetic data is then used to train $12$ classifiers (see Table \ref{tab:classification_data_dep}), which are tested on remaining $80\%$ of input data.

For training DP-MERF synthesizers, we set parameters as in \cite{DBLP:conf/aistats/HarderAP21}, i.e. number of epochs $4000$, number of Furier features $5000$, mini-batch side $0.5$, undersampling rate $0.005$. For our data dependent algorithm, we use undersampling rate of $0.005$, and set the number of data independent levels to be equal to $30$ and maximal number of levels to $60$.

As comparison metrics, we use ROC (area under the receiver operating curve). Table \ref{tab:classification_data_dep} shows average ROC for our data dependent algorithm over $20$ repetitions for each classifier, as well as average ROC over the classifiers. Table \ref{tab:classification_dp_merf} shows average ROC over classifiers for DP-MERF with $5$ repetitions for each classifier. 

Although our data dependent algorithm does not outperform DP-MERF, its performance degrades slower for increasing privacy.
\begin{table}[hb]
    \centering
    \caption{Data dependent algorithm: ROC for various levels of privacy. Average over 20 repetitions.}\label{tab:classification_data_dep}
    \begin{tabular}{ c c c c c }
      \toprule % from booktabs package
      &\bfseries $\epsilon=10$ & \bfseries $\epsilon=1$ & \bfseries $\epsilon=0.1$ & \bfseries $\epsilon=0.01$\\
      \midrule % from booktabs package
      %\bfseries Logistic Regression & 0.12345\\
      \bfseries Logistic Regression & 0.705 & 0.545 &0.481 &0.527  \\ 
  \bfseries Gaussian Naive Bayes & 0.562 & 0.563 & 0.479 & 0.547  \\
  \bfseries Bernoulli Naive Bayes & 0.495 & 0.564 & 0.497 & 0.521  \\
\bfseries  Linear SVM & 0.758 & 0.524 & 0.508 & 0.546 \\
\bfseries  Decision Tree &  0.676 & 0.611 & 0.519 & 0.532  \\ 
\bfseries LDA &  0.518 & 0.580 & 0.480 & 0.542  \\
\bfseries Ada Boost &  0.632 & 0.572 & 0.485 & 0.521  \\
\bfseries Bagging &  0.673 & 0.579 & 0.518 & 0.508  \\
\bfseries Random Forest & 0.663 & 0.594 & 0.530 & 0.543  \\
\bfseries GBM &  0.631 & 0.582 & 0.521 & 0.523 \\
\bfseries Multi-layer percepton & 0.625 & 0.553 & 0.486 & 0.525  \\
\bfseries XGBoost &  0.588 & 0.598 & 0.527 & 0.478\\
\bfseries Average &   0.627 & 0.572 & 0.503 & 0.526 \\
      \bottomrule % from booktabs package
    \end{tabular}
\end{table}

\begin{table}[hb]\label{table appendix DP MERF}
    \centering
    \caption{DP-MERF: ROC for various levels of privacy. Average over 5 repetitions.}\label{tab:classification_dp_merf}
    \begin{tabular}{ c c c c }
      \toprule % from booktabs package
      &\bfseries $\epsilon=10$ & \bfseries $\epsilon=1$ & \bfseries $\epsilon=0.1$ \\
      \midrule % from booktabs package
      %\bfseries Logistic Regression & 0.12345\\
      \bfseries Average &   0.880 & 0.792 & 0.564 \\
      \bottomrule % from booktabs package
    \end{tabular}
\end{table}

\clearpage

\section{Approximating Gaussian Distribution by Mixture of Uniforms}\label{appendix:proxy_mixture_uniforms}

\subsection{MMD}
Let us assume that the data is arising from a Gaussian distribution and we are estimating with the samples $\{z_i\}_1^n$. The maximum-mean discrepance (MMD) between the population P and the samples is given by: 
\begin{equation}
\mathrm{MMD}_{u}^{2}(\mathcal{N}_{d},Q_{n})=E_{x,x'\sim\mathcal{N}_{d}}[k(x,x')]-\frac{2}{n}\sum_{i=1}^{n}E_{x\sim\mathcal{N}_{d}}[k(x,z_{i})]+\frac{1}{n(n-1)}\sum_{i=1}^{n}\sum_{j\neq i}^{n}k(z_{i},z_{j}).\label{eq:mmdu}
\end{equation}
The expectations in the expression above can be computed analytically
to yield the formula~\cite{https://doi.org/10.1002/sta4.329}:

\[
\mathrm{MMD}_{u}^{2}(\mathcal{N}_{d},Q_{n})=\left(\frac{\gamma^{2}}{2+\gamma^{2}}\right)^{d/2}-\frac{2}{n}\left(\frac{\gamma^{2}}{1+\gamma^{2}}\right)^{d/2}\sum_{i=1}^{n}e^{-\frac{\Vert z_{i}\Vert^{2}}{2(1+\gamma^{2})}}+\frac{1}{n(n-1)}\sum_{i=1}^{n}\sum_{j\neq i}^{n}e^{-\frac{\Vert z_{i}-z_{j}\Vert^{2}}{2\gamma^{2}}}.
\]

\subsection{Widths of bins}
Let us assume that we have $2k+1$ boxes to approximate the Gaussian distribution where we assume an odd number of boxes to apply symmetry arguments. 
The distributions are given by:
\begin{align}
    P(x) = \frac{1}{\sqrt{2\pi}} e^{-x^2/2}
    Q(x) = w_0 I_0 + \sum_{i=1}^k w_i I_i + \sum_{i=1}^k w_{-i} I_{-i}
\end{align}
where $I_i = I[(2i-1)C <= x < (2i+1)C]$ and $\sum_{i=-k}^{k} w_i = 1$.
The KL divergence between distributions is given by:
\begin{align}
    D_{KL}(P||Q) := -\int_{\infty}^{\infty} p(x) \log\frac{p(x)}{q(x)} dx
\end{align}
and in particular for our setting is given by:
\begin{align}
   D_{KL}(Q||P) = \sum_{i=-k}^k w_i \log\frac{w_i}{2C} + \frac{1}{2} \log(2\pi) + \sum_{i=-k}^k \frac{w_i C^2}{6} (12i^2+1)  
\end{align}
Using the KL bounds for measuring divergence, we are able to obtain the following weights and size of the boxes. Each box is of size given by $2c$ and placed at location $2ci$ for $i\in[-k,k]$. 
Applying the method of Lagrange multipliers, we can obtain
the optimal box weights for a mixture of uniform distributions with respect to the Gaussian distribution:
\begin{align}
w_0 &= \frac{1}{1 + 2\sum_{i=1}^k e^{-2i^2c^2}} \\
w_i &= w_0 e^{-2i^2 c^2} \qquad \forall i \in [-k, k] \textrm{ and } i\ne 0
\end{align}

\end{document}